\documentclass[journal,10pt]{IEEEtran}
\usepackage{blindtext}
\usepackage{graphicx}
\usepackage{amsmath}
\usepackage{graphicx}
\usepackage{amsthm}
\usepackage{multirow}
\usepackage{amsmath}

\DeclareMathOperator*{\argmin}{argmin}
\usepackage{yhmath}
\usepackage{amssymb}
\usepackage{array}
\usepackage{leftidx}
\usepackage{longtable}
\usepackage{stfloats}
\usepackage{cite}
\usepackage{url}
\usepackage{subfigure}
\usepackage{comment}
\usepackage{color}
\usepackage{algorithm}
\usepackage{algorithmic}
\usepackage{lipsum}



\begin{document}
\theoremstyle{plain}
\newtheorem{thm}{Theorem}
\newtheorem{remark}{Remark}
\newtheorem{lemma}{Lemma}
\newtheorem{prop}{Proposition}
\newtheorem*{cor}{Corollary}
\theoremstyle{definition}
\newtheorem{defn}{Definition}
\newtheorem{condi}{Condition}
\newtheorem{assump}{Assumption}

\title{Safe and Human-Like Autonomous Driving: A Predictor-Corrector Potential Game Approach}
\author{Mushuang Liu\IEEEauthorrefmark{1},~\IEEEmembership{Member,~IEEE,}  ~H. Eric Tseng\IEEEauthorrefmark{3},  ~Dimitar Filev\IEEEauthorrefmark{3},~\IEEEmembership{Fellow,~IEEE,} ~ Anouck Girard\IEEEauthorrefmark{2},~\IEEEmembership{Senior Member,~IEEE} and Ilya Kolmanovsky\IEEEauthorrefmark{2},~\IEEEmembership{Fellow,~IEEE}

\thanks{\IEEEauthorrefmark{1} M. Liu is with the Department of Mechanical and Aerospace Engineering, University of Missouri, Columbia, MO, USA (email: ml529@missouri.edu).
}

\thanks{\IEEEauthorrefmark{3} H. E. Tseng, and D. Filev are with Ford Research and Innovation Center, 2101 Village Road, Dearborn, MI 48124, USA (e-mail: htseng@ford.com and dfilev@ford.com).}
\thanks{\IEEEauthorrefmark{2} A. Girard and I. Kolmanovsky are with the Department of  Aerospace Engineering, University of Michigan, Ann Arbor, MI, USA (email: anouck@umich.edu and ilya@umich.edu).
}
\thanks{This work is supported by Ford Motor Company.}
}
\maketitle
\markboth{IEEE Transactions on Control Systems Technology} {Liu \MakeLowercase{\textit{et al.}}: Potential Game Based  Frameworks for Decision-Making in Autonomous Driving}
\begin{abstract}
This paper proposes a novel  decision-making framework for autonomous vehicles (AVs), called predictor-corrector potential game (PCPG), composed of a Predictor and a Corrector. To enable human-like reasoning and characterize agent interactions, a receding-horizon multi-player game is formulated. To address the challenges caused by the complexity in solving a multi-player game and by the requirement of real-time operation, a  potential game (PG) based decision-making framework is developed in the PG Predictor, where the agents' cost functions are heuristically predefined.
We acknowledge that the behaviors of other traffic agents, e.g., human-driven vehicles and pedestrians,  may not necessarily be consistent with the predefined cost functions. To address this issue, a best response based PG Corrector is designed. In the Corrector, the action deviation between the ego vehicle prediction and the surrounding agents' actual behaviors are measured and are fed back to the ego vehicle decision-making, to correct the prediction errors caused by the inaccurate predefined cost functions and to improve the ego vehicle strategies.  

Distinguished from most existing game-theoretic approaches, this PCPG  1) deals with multi-player games and guarantees  the existence of a pure-strategy Nash equilibrium (PSNE) and the convergence of the PSNE seeking algorithm; 2) is computationally scalable in a multi-agent scenario; 3) guarantees the ego vehicle safety under suitable conditions; and 4)  approximates the actual PSNE of the system  despite the unknown cost functions of others. Comparative studies between the PG, the PCPG, and the control barrier function (CBF) based approaches are conducted in diverse traffic  scenarios, including oncoming traffic scenario  and multi-vehicle intersection-crossing scenario. The results from validation case studies based on a naturalistic dataset are reported.     
\end{abstract}
\section{Introduction}

Fully autonomous vehicles (AVs) {are expected to improve} safety, mobility, accessibility, and reduce energy consumption \cite{online2}. {However, their broad introduction still faces significant challenges} \cite{online3}. One of the key technical challenges lies in the design of AV decision-making algorithms, which aim to generate safe, reliable, and intelligent decisions for AVs in diverse and complex traffic scenarios. 
An ideal AV decision-making framework  is expected to have the following properties.
\begin{enumerate}
    \item \textit{Safety}: Theoretical safety guarantees are desirable, if a safe solution exists;
    \item \textit{Interpretability}: The decision-making process is expected to be interpretable {by} humans to build trust between humans and AVs;
    \item \textit{Applicability}: The algorithm should be applicable to various traffic scenarios and {have} the capability  to handle unfamiliar or even previously unseen situations;
    \item \textit{Scalability}: The algorithm is computationally scalable to handle a large number of traffic agents;
    \item \textit{Intelligence}:  Human-like negotiating behaviors and reasoning are desirable, especially when AVs interact with human drivers and/or pedestrians.
\end{enumerate}

The existing approaches to AV decision-making can be largely grouped into two categories: model-based and data-driven. Model-based approaches \cite{RSS,cbf,logical_1,logical_2,MDP,fuzzy}, in general, have good explainability  and have the potential to provide  safety guarantees under certain conditions ({this is the case for} responsibility-sensitive safety (RSS) \cite{RSS} and control barrier function (CBF) \cite{cbf} -based approaches). However, the downside is that they often {lead to} conservative AV behaviors \cite{rss_conservative}, are highly dependent on assumed model parameters\cite{RSS_parameter}, and lack the applicability to diverse traffic scenarios \cite{applicability}. Data-driven approaches, on the other hand, take advantage of naturalistic traffic datasets \cite{Ford_1,14,15} and have the capability to generate human-like behaviors \cite{humanlike_learning} in certain scenarios. However, their performance {is highly reliant} on the quality of the training datasets, lacking the assurance in dealing with unfamiliar or unseen situations.  Moreover, since deep neural networks (DNNs) are usually employed, especially in the end-to-end learning \cite{endtoend,endtoend2}, {the lack of transparency in the decision-making process becomes a concern in achieving intrinsic interpretability \cite{explainable_2} and in building trust and confidence in AVs \cite{explainable_1}.} 
In addition, theoretical safety guarantees are generally not available in purely data-driven approaches.

Game-theoretic approaches have the potential to combine the advantages of the model-based and the data-driven {approaches}  \cite{human,my_game,my_graphicalgame,Victor,mine_1,online_payoff}. Given the cost function of each agent, a game-theoretic decision-making  is consistent with human reasoning: Human behaviors are naturally motivated by {pursuing} their own interests {while accounting for interactions with others} \cite{human}. On the other hand, data-driven approaches can {facilitate} the AV cost function design, leading to human-like behaviors. Examples include  supervised learning, reinforcement learning (RL) with value function approximation, and inverse RL \cite{supervisedlearning,inverse_RL,RL_VFA}. However, conventional game-theoretic approaches often suffer from scalability issues and from the lack of knowledge {of} the surrounding agents' cost functions that reflect the variability in human driving behaviors. 

To address the scalability challenge, pairwise games have been widely-adopted in the literature \cite{Victor,deadlock}. {In these kinds of approaches}, the ego vehicle is {assumed} to play multiple $2$-player games, instead of one multi-player game, and the most conservative outcomes are selected as the final decision. Such pairwise games {may} lead to  conservative AV behaviors and  result in {a} deadlock even in an ideal environment where all agents employ the pairwise games \cite{deadlock}. 

To address the challenge {of} unknown surrounding agents' cost functions,  attempts have been made to learn the parameters of the cost function in real time \cite{dimitar_dynamics_2,dimitar_dynamics_3,online_payoff}. These parameters usually characterize driving style \cite{supervisedlearning}, aggressiveness \cite{online_payoff}, or social value orientation (SVO) \cite{sov,closedloop1}. 
However, because the time duration of traffic agents' interaction is usually short, the amount of data {may not be} sufficient to guarantee the online learning performance. An inaccurate {estimate of the} cost can mislead the ego vehicle to perform undesirable or even dangerous {maneuvers} \cite{mine_1}.

To address the above limitations and enable the AV decision-making to meet the expectations listed in the first paragraph, we propose a novel predictor-corrector potential game (PCPG) framework, composed of two main components: a PG Predictor and a PG Corrector. The Predictor solves a multi-player game with predefined cost functions to find the optimal strategy for the ego vehicle while taking into consideration the interactions with other traffic agents. 
The Corrector aims to correct the prediction error caused by the inaccurate predefined cost functions of others and to improve the ego vehicle decision-making.  
With this PCPG:  
\begin{itemize}
    \item the existence of a  pure-strategy Nash equilibrium (PSNE) and the convergence of the solution seeking algorithm are guaranteed;
    \item the computational scalability challenge is addressed; 
    \item the ego vehicle safety is guaranteed under suitable conditions. These conditions are detailed in Section V; 
    \item the optimal solution for the ego vehicle can be approximated despite unknown cost functions of others. 
\end{itemize} 

To summarize, this PCPG framework  inherits the appealing properties of game-theoretic approaches, including human-like reasoning and explicit agent interaction characterization, {and} addresses the computational scalability and  lack of information challenges, improving the  practicability and reliability of applying game-theoretic approaches to autonomous driving. 

This paper is organized as follows. Section \ref{preliminary} introduces preliminaries to facilitate the analysis in this paper. Section \ref{problem_formulation} formulates the AV decision-making problem as a receding-horizon multi-player game problem. Section \ref{PCPG} proposes the PCPG. Section \ref{performance} analyzes the  PCPG performance, including safety and optimality. Section \ref{simulation} conducts numerical studies in specific traffic scenarios and reports validations using a naturalistic dataset. Section \ref{conclusion} concludes the paper.  
  
\section{Preliminaries}\label{preliminary}

Let $\mathbb{Z}_{+}(\mathbb{Z}_{++})$  denote the set of non-negative (positive) integers and $\mathbb{R}_{+}(\mathbb{R}_{++})$  denote the set of non-negative (positive) real numbers. 

\begin{defn} [Lipschitz continuous function \cite{book_lf}] \label{d_lf} 
Suppose $(\mathcal{Y},d_{\mathcal{Y}})$ and $(\mathcal{Z},d_{\mathcal{Z}})$ are metric spaces and $g:\mathcal{Y}\rightarrow \mathcal{Z}$. If there exists $K\in\mathbb{R}_{+}$ such that 
\begin{equation}
    d_{\mathcal{Z}}\left(g(y_1),g(y_2)\right)\leq K\cdot d_{\mathcal{Y}}(y_1,y_2) \quad \forall y_1,y_2\in \mathcal{Y},
\end{equation} then $g$ is called a Lipschitz continuous function on $\mathcal{Y}$ with Lipschitz constant $K$.
\end{defn}

To introduce the preliminaries on game theory, let us
consider a strategic-form game $\mathcal{G}=\{\mathcal{N},\mathcal{A},\{J_i\}_{i\in\mathcal{N}}\}$. Here $\mathcal{N}=\{1,2,...,N\}$ is the set of players (or agents), $\mathcal{A}=\mathcal{A}_1\times \mathcal{A}_2\times \dots \times \mathcal{A}_N$ with $\mathcal{A}_i$ representing the strategy space of player $i$, and $J_i: \mathcal{A}\to \mathbb{R}$ is the cost function of player $i$. Denote {by} $\mathcal{N}_{-i}$ the set of all agents except for agent $i$. We let $\mathbf{a}_i\in\mathcal{A}_i$ represent the strategy of agent $i$ and $\mathbf{a}_{-i}\in\mathcal{A}_{-i}$ represent the set of strategies of all other agents except for agent $i$, i.e., $\mathbf{a}_{-i}=\{\mathbf{a}_1,\dots,\mathbf{a}_{i-1},\mathbf{a}_{i+1},\dots,\mathbf{a}_N\}$ with $\mathcal{A}_{-i}$ being the  domain of $\mathbf{a}_{-i}$. Denote $\mathbf{a}=\{\mathbf{a}_i,\mathbf{a}_{-i}\}\in\mathcal{A}$. Similarly, let $J_{-i}=\{J_1,\dots,J_{i-1},J_{i+1},\dots,J_N\}$ and $J=\{J_i,J_{-i}\}$. 

\begin{defn} [Best Response \cite{game_book}] \label{d1} 
Agent $i$'s best response to other agents' fixed strategies $\mathbf{a}_{-i}\in\mathcal{A}_{-i}$ is defined as the strategy $\mathbf{a}_i^*$ such that
\begin{equation}
    J_i( \mathbf{a}_i^*, \mathbf{a}_{-i})\leq J_i(\mathbf{a}_i, \mathbf{a}_{-i}) \quad \forall \mathbf{a}_i\in\mathcal{A}_i.
\end{equation}
\end{defn}

\begin{defn} [Pure-Strategy Nash Equilibrium \cite{game_book}] \label{d2}
An $N$-tuple of strategies (or strategy profile) $\{\mathbf{a}_1^*,\mathbf{a}_2^*,...,\mathbf{a}_N^*\}$  is  a pure-strategy Nash equilibrium for an $N$-player game if and only if   
\begin{equation}\label{Nash_de}
 J_i(\mathbf{a}_i^*,\mathbf{a}_{-i}^*)\leq J_i(\mathbf{a}_i,\mathbf{a}_{-i}^*)\quad  \forall \mathbf{a}_i\in\mathcal{A}_i, \forall i\in\mathcal{N}.
\end{equation}
\end{defn}
Equation \eqref{Nash_de} implies that  if all agents play their best response, then a PSNE is achieved, and if a PSNE is achieved, then no player would have the incentive to change its strategy. 

Next we define a special class of games, called \textit{Continuous Potential Game}. 
{
\begin{assump}\label{As1}
$\mathcal{A}_i$ is a connected set, $\mathcal{A}_i\neq \emptyset$, $\forall i\in\mathcal{N}$, and $\mathcal{A}$ is compact, i.e., closed and bounded.
\end{assump}}

\begin{defn} [Continuous Potential Game \cite{my_potential}] \label{d5}
{Under Assumption \ref{As1},} let $J_i$ be everywhere differentiable  on an open superset of $\mathcal{A}$. The game $\mathcal{G}$
is called a continuous potential game if and only if there exists a {function}  $F:\mathcal{A}\rightarrow\mathbb{R}$ such that $F$ is  everywhere differentiable  on an open superset of $\mathcal{A}$, and 
\begin{equation}\label{defi_c}
    \frac{ \partial J_i(\mathbf{a}_i,\mathbf{a}_{-i})}{\partial \mathbf{a}_i}=\frac{\partial F(\mathbf{a}_i,\mathbf{a}_{-i})}{\partial \mathbf{a}_i}
\end{equation}
holds $\forall \mathbf{a}_i\in\mathcal{A}_i$,  $\forall \mathbf{a}_{-i}\in\mathcal{A}_{-i}$, and  $\forall i\in\mathcal{N}$. The function $F$ is called the potential function.
\end{defn}

Throughout this paper, the term PG  refers to a continuous PG, and the word ``continuous" may be omitted. 

A PG has appealing properties. We {summarize} two of them in the following lemmas. 

\begin{lemma}\label{l1}[Existence of PSNE \cite{potential_book, my_potential}]  {Under  Assumption \ref{As1}}, if $\mathcal{G}$ is a PG, then it has  at least one PSNE. Moreover, if the potential function $F$ is strictly convex, then the PSNE is unique. 
\end{lemma}

\begin{lemma}\label{l2}[Equivalence of Nash {equilibria} sets \cite{potential_book, my_potential}]
{Under  Assumption \ref{As1}}, if $\mathcal{G}$  is a PG, then the set of PSNE coincides with the set of PSNE of the identical-interest game $\mathcal{G}'=\{\mathcal{N},\mathcal{A},\{F\}_{i\in\mathcal{N}}\}$ with $F$ being the potential function. That is,
\begin{equation}
    \text{NESet}(\mathcal{G})=\text{NESet}(\mathcal{G}'),
\end{equation}
where $\text{NESet}$ denotes the set of PSNE.
\end{lemma}
\section{Problem formulation}\label{problem_formulation}
Consider a set of traffic agents $\mathcal{N}$ {represented by the following discrete time models:}
\begin{equation}\label{dynamics}
    x_i(t+1)=f_i(x_i(t),a_i(t)),
\end{equation}
where  $x_i(t)\in\mathcal{X}_i\subseteq \mathbb{R}^{n_i}$ and $a_i(t)\in\mathcal{U}_i\subseteq \mathbb{R}^{m_i}$ are, respectively, the state and action of agent $i$ at the time step $t$, $f_i$ is the system evolution model, and $i\in\mathcal{N}$. Denote $x_{-i}=\{x_1,\dots,x_{i-1},x_{i+1},\dots,x_N\}$, $x=\{x_i,x_{-i}\}$, $a_{-i}=\{a_1,\dots,a_{i-1},a_{i+1},\dots,{a}_N\}$, $a=\{a_i,a_{-i}\}$, and $f=\{f_1,f_2,\dots,f_N\}$. Denote the dimension of $a_{-i}$ as $m_{-i}$, i.e., $a_{-i}(t)\in \mathbb{R}^{m_{-i}}$. 

In autonomous driving applications, $x_i$ usually contains agent $i$'s  position, velocity, and heading angle information, and $a_i$ is the acceleration and/or the angular velocity. These signals can be measured by/estimated from vehicle-mounted sensors and have been used in driver assistance systems implemented in real vehicles \cite{tesla}. The system evolution model $f_i$ represents agent dynamics. Examples of $f_i$ include single-mass model \cite{add_dynamics_RL,add_dynamics_nan}, unicycle model \cite{add_dynamics_RTD}, and bicycle model \cite{bicycle}.
\begin{assump}
At time $t$, agent $i$ has access to $x(t)$, $a(t-1)$ and $f$.
\end{assump}

In a driving scenario, every traffic agent has its own driving objective, e.g., tracking a desired trajectory without {collisions} and with ride comfort. We use the cost function $J_i$ to characterize agent $i$'s objective and formulate the decision-making problem as a receding horizon optimal control problem. That is, at each $t$, agent $i$ aims to find its optimal action sequence (also called strategy, or behavior) $\mathbf{a}^*_i(t)$ such that
\begin{equation}\label{optimization}
    \begin{split}
    \mathbf{a}^*_i(t)&\in \argmin_{\mathbf{a}_i(t)\in\mathcal{A}_i}J_i(\mathbf{a}_i(t),\mathbf{a}_{-i}(t))
    \\
    &
    =\argmin_{\mathbf{a}_i(t)\in\mathcal{A}_i} \sum_{\tau=t}^{t+T-1}\Psi_i(x_i(\tau),x_{-i}(\tau),a_i(\tau),a_{-i}(\tau)),
    \end{split}
\end{equation}
where $\mathbf{a}_i(t)=\{a_i(t),a_i(t+1),\dots,a_i(t+T-1)\}\in\mathcal{A}_i$, $\mathcal{A}_i$ is determined by $\mathcal{U}_i$, $\Psi_i$ is the instantaneous cost at one time instant,  $T\in\mathbb{Z}_{++}$ is the horizon length, and $J_i:\mathcal{A}\rightarrow\mathbb{R}$ is agent $i$'s cost function. The expression of $J_i(\mathbf{a}_i(t),\mathbf{a}_{-i}(t))$ can be different at different $t$ due to its state dependence. With a slight abuse of notation, we do not distinguish $J_i$ at different $t$ in this paper, as it is clear from the context. After deriving $\mathbf{a}^*_i(t)$, agent $i$ implements the first element
 $a^*_i(t)$ and repeats the same procedure at the next time instant, $t+1$.

In \eqref{optimization}, $J_i$ depends on both $\mathbf{a}_i$ and $\mathbf{a}_{-i}$, indicating that agent $i$'s driving performance is affected by not only its own actions but also the actions of other traffic agents, characterizing agent interactions. At each $t$, if every {agent} $i$ aims to optimize its performance $J_i$, then the multi-agent optimization \eqref{optimization} becomes a multi-player game, and the set of all agents' optimal strategies $\{\mathbf{a}_1^*(t),\mathbf{a}_2^*(t),...,\mathbf{a}_N^*(t)\}$, if nonempty, composes a PSNE according to Definition \ref{d2}.

\begin{remark}
The receding horizon game \eqref{optimization} is consistent with human driver decision-making in the following sense. 
\begin{enumerate}
    \item Human driver actions are generally motivated by a foreseen gain or loss within a finite prediction horizon. For example, a braking maneuver is usually motivated by a foreseen collision in a few seconds if the current speed is continued, consistent with the predictive control {setting of} \eqref{optimization}.
    \item Human drivers adjust their strategies frequently according to the latest information. For example, if  a human driver is cut off by a car,  he/she may have to abort the previous strategy (e.g., keeping speed) and take actions to avoid collision. Such real-time adjustment is achieved in \eqref{optimization} {since} a new strategy is planned at every $t$. 
    \item Human driver strategies are  affected by {the} surrounding traffic agents. An experienced human driver does not simply react to the actions of others, instead, he/she  predicts others' behaviors by considering their driving objectives/interests, to benefit his/her own decision-making. This feature is represented in \eqref{optimization} by the coupled agent cost functions and by the game-theoretic formulation.   
\end{enumerate}
Such optimization-based  approaches to represent human decision-making during driving have been validated in the literature with naturalistic datasets \cite{validation_1,validation_2,validation_3,validation_4,payoff1,closedloop1}. 
\end{remark}

The {practical use of the} formulation \eqref{optimization}, although consistent with human reasoning, {presents several challenges}:
\begin{itemize}
    \item \textbf{\textit{Existence of solution:}} Given arbitrary $J_i$,  a strategy profile $\{\mathbf{a}_1^*(t),\mathbf{a}_2^*(t),...,\mathbf{a}_N^*(t)\}$ that satisfies   \eqref{optimization} for all $i\in\mathcal{N}$ may not always exist.
    \item \textbf{\textit{Convergence of algorithm:}} Even if a solution exists, a solution seeking algorithm, e.g., best- and better- response dynamics, may not necessarily converge.
    \item \textbf{\textit{Nonuniqueness of PSNE:}} The game \eqref{optimization} may have multiple solutions, i.e., multiple PSNE.  
    {Selecting the one that is preferable to others can also be challenging.}  
    \item \textbf{\textit{Computational scalability:}} Solving \eqref{optimization} generally requires multiple and iterative optimizations, resulting in high computational burden when $N$ is large.
    \item \textbf{\textit{Lack of information:}} To solve \eqref{optimization}, the ego vehicle needs to know everyone's cost function, which  {may not be realistic in a traffic setting}.  
\end{itemize}

To differentiate the ego vehicle and its surrounding agents, from now on, we use agent $i$ to represent the ego vehicle, and the set $\mathcal{N}_{-i}$ to represent its surrounding agents, which may include AVs, human-driven vehicles, and pedestrians. 

\begin{remark}
Although agent $i$ does not know $J_{-i}$, it may have access to a typical human driver or pedestrian cost function, calibrated using naturalistic traffic datasets \cite{inverse_RL,payoff1}. Such a typical cost function corresponds to ``common" or ``general" human behaviors.  However, given a specific human driver, his/her behaviors may not necessarily be consistent with  the predetermined typical cost function, as human drivers often have distinct driving habits and styles, and their behaviors can hardly be described by one universal cost function. 
\end{remark}

Let $\hat{J}_j: \mathcal{A}\to \mathbb{R}$ be a predefined typical cost function for agent $j$.  $\hat{J}_j$ may not be equal to $J_j$. 

This paper aims to develop a decision-making framework to solve \eqref{optimization} and to address the aforementioned challenges. 

\begin{figure*}
    \centering
    \includegraphics[width=0.8\textwidth]{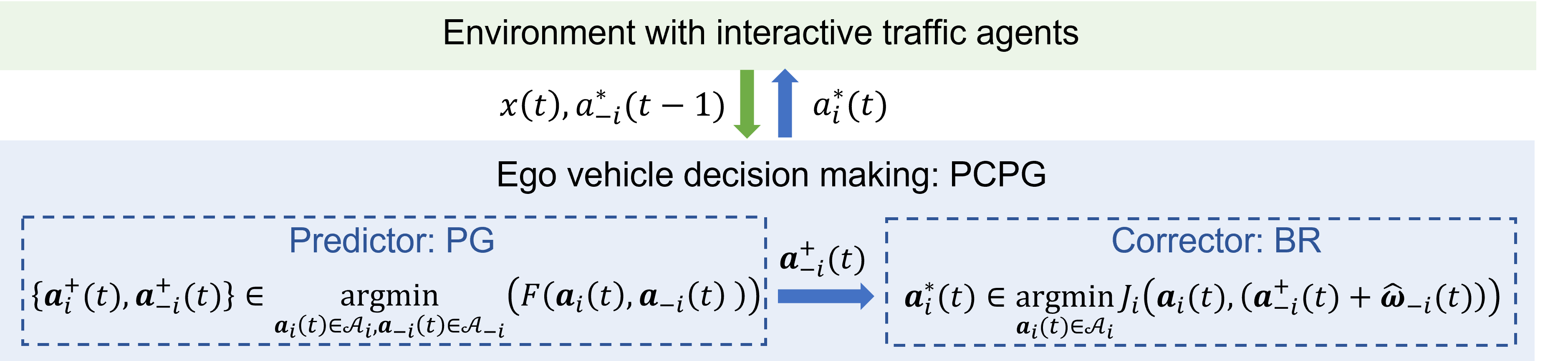}
    \caption{Diagram of the predictor-corrector potential game at time $t$. }
    \label{fig1}
\end{figure*}

\section{Predictor-Corrector Potential Game}\label{PCPG}
This section proposes a predictor-corrector potential game {approach} to solve \eqref{optimization}. The predictor-corrector structure is inspired by a previous work on CBF-based collision avoidance \cite{pcca}. Here we apply this structure in a game-theoretic setting to address the challenges caused by unknown cost functions of other agents. The Predictor {solves}  \eqref{optimization} with the predefined cost function $\hat{J}_j$, and the Corrector  {corrects} the prediction error caused by $\hat{J}_j\neq {J}_j$. The  diagram of the PCPG framework is shown in Figure \ref{fig1}.  


\subsection{Predictor}
In the Predictor step, the ego vehicle assumes that the surrounding agents' behaviors are generated from $\hat{J}_j$. Therefore, it solves the following $N$-player game at each $t$, 
\begin{equation}\label{predictor}
    \begin{split}
    \mathbf{a}^+_j(t)&\in \argmin_{\mathbf{a}_j(t)\in\mathcal{A}_j}\hat{J}_j(\mathbf{a}_j(t),\mathbf{a}_{-j}(t))
    \end{split}
\end{equation}
where $j\in\mathcal{N}$, $\mathbf{a}^+_j(t)$ is the predicted agent $j$'s action sequence at $t$, and $\hat{J}_i=J_{i}$. 

\begin{assump}\label{A2}
$\hat{J}_j$ is  everywhere differentiable  on an open superset of $\mathcal{A}$, $\forall j\in\mathcal{N}$.
\end{assump}

To ensure a PSNE always exists and is obtainable in real time, we formulate the game \eqref{predictor} as a PG. 

\begin{thm}[Theorem 6 in \cite{my_potential}]\label{t1}
If the cost function $\hat{J}_j$ in \eqref{predictor} satisfies
\begin{equation}\label{cost_form}
\begin{split}
    &\hat{J}_j(\mathbf{a}_j(t),\mathbf{a}_{-j}(t))\\
    &=\alpha J_j^{self}(\mathbf{a}_j(t))+\beta \sum_{k\in\mathcal{N},k\neq j} J_{jk}(\mathbf{a}_j(t),\mathbf{a}_{k}(t)),
\end{split}
\end{equation}
where  $J_j^{self}: \mathcal{A}_j\to \mathbb{R}$ is a function {determined} solely by agent $j$'s action, $J_{jk}: \mathcal{A}_j\times\mathcal{A}_k\to \mathbb{R}$ satisfies
\begin{equation}
\begin{split}
  &\quad  J_{jk}(\mathbf{a}_j(t),\mathbf{a}_{k}(t))=J_{kj}(\mathbf{a}_k(t),\mathbf{a}_{j}(t)), \\
  &\forall j,k\in\mathcal{N}, j\neq k, \text{ and } \forall \mathbf{a}_j\in \mathcal{A}_j, \mathbf{a}_{k}\in \mathcal{A}_k,
    \end{split}
\end{equation}
and $\alpha$ and $\beta$ are two {real numbers}. Then the game \eqref{predictor} is a PG with the following potential function, 
\begin{equation}\label{potential_design}
\begin{split}
        & F(\mathbf{a}(t))\\
    & =\alpha\sum_{j\in\mathcal{N}}J_j^{self}(\mathbf{a}_j(t))+\beta\sum_{j\in \mathcal{N}}\sum_{k\in\mathcal{N},k< j} J_{jk}(\mathbf{a}_j(t),\mathbf{a}_{k}(t)).
\end{split}
\end{equation}
\end{thm}
Theorem \ref{t1} states that if $\hat{J}_j$ can be {represented}  as a linear combination of two components: $J_j^{self}$ and  $\sum_{k\in\mathcal{N},k\neq j}J_{jk}$, then the resulting game is a PG. The required cost function form, i.e., \eqref{cost_form}, meets the autonomous driving application needs in general. Specifically, the first component $J_j^{self}$ can  model self-focused objectives, including tracking a desired speed, minimizing fuel consumption, and maintaining ride comfort. The second component  $\sum_{k\in\mathcal{N},k\neq j}J_{jk}$ can characterize symmetric pairwise agent interaction, such as a pairwise collision penalty.  See \cite{payoff1,payoff2,payoff3,dynamics} for examples where AV cost function design follows, or can be {slightly} revised to follow, the form \eqref{cost_form}, and \cite{PCPG_add1,PCPG_add2,my_potential} for examples where potential games are employed in the context of driving. 

After formulating the game as a PG, according to Lemma \ref{l1}, a PSNE always exists.  According to Lemma \ref{l2}, a global minimizer to the following optimization problem provides a PSNE to the game \eqref{predictor},
\begin{equation}\label{potential}
    \min_{\mathbf{a}(t)\in\mathcal{A}} F(\mathbf{a}(t)),
\end{equation}
where $F:\mathcal{A}\rightarrow\mathbb{R}$ is the potential function determined by \eqref{potential_design}.
Denote a global minimizer to \eqref{potential} as $\mathbf{a}^+(t)$, i.e., 
  $  \mathbf{a}^+(t)\in\argmin_{\mathbf{a}(t)\in\mathcal{A}} F(\mathbf{a}(t))$.

With \eqref{potential}, the PSNE seeking problem is transferred to a {simpler} optimization problem. 
To see the difference clearly, we write down both algorithms here. Algorithm \ref{A_br} is the best-response dynamics, which is among the most widely used algorithms to solve a multi-player game. As shown in \cite{BR}, the number of required optimizations in this algorithm can increase exponentially with the number of game players, thus not scalable. Algorithm \ref{A1} is the potential function optimization, applicable if \eqref{predictor} is a PG. This algorithm only requires one optimization, regardless of the number of game players. While this optimization problem is with a larger number of players, its size grows linearly with the number of players. The comparison of the computational times of the two algorithms in the context of 5-vehicle intersection-crossing is reported in \cite{my_potential}, showing computational advantages of PG approach. In addition, Algorithm \ref{A_br} does not always converge even if a PSNE exists, i.e., the \textit{NashCondition} in Algorithm \ref{A_br} may never become \textit{True}. If the game is a PG, then Algorithm \ref{A_br} is guaranteed to converge \cite{my_potential}.

\begin{algorithm}[t]
\caption{Best response dynamics to solve \eqref{predictor}} \label{A_br}
\hspace*{0.0in} {\bf Inputs:} \\ 
\hspace*{0.2in}Agent set  $\mathcal{N}$;  \\
\hspace*{0.2in}System state  $x(t)$;  \\
\hspace*{0.2in}System dynamics  $f$;\\
\hspace*{0.2in}Strategy space  $\mathcal{A}$;  \\
\hspace*{0.2in}Cost functions  $\hat{J}_j$, $j\in\mathcal{N}$;\\
\hspace*{0.02in} {\bf Output:} \\
\hspace*{0.2in} PSNE $\mathbf{a}^+(t)$.\\
\hspace*{0.02in} {\bf Procedures:} 
\begin{algorithmic}[1]
\STATE Set \textit{NashCondition=False}
\STATE {\bf While} \textit{NashCondition=False} {\bf do}
\STATE \hspace*{0.2in}{\bf For} $j=1,2,...,N$ {\bf do}
\STATE \hspace*{0.3in} Find $\mathbf{a}_j^+(t)$ according to 
\hspace*{0.3in} 
\begin{equation}\nonumber
    \quad\quad\quad\mathbf{a}_j^+(t)\in\argmin_{\mathbf{a}_j(t)\in\mathcal{A}_j} \hat{J}_j(\mathbf{a}_j(t),\mathbf{a}_{-j}(t)).
\end{equation}
\STATE \hspace*{0.3in} Update $\mathbf{a}_j(t)$ using $\mathbf{a}^+_j(t)$.
\STATE \hspace*{0.2in} {\bf End for}
\STATE \hspace*{0.2in}{\bf If} \\
\STATE  \hspace*{0.3in} $
    \mathbf{a}^+_j(t)\in\argmin\limits_{\mathbf{a}_j(t)\in\mathcal{A}_j} \hat{J}_j(\mathbf{a}_j(t),\mathbf{a}_{-j}^+(t))$ holds    $\forall j\in\mathcal{N}$,
\STATE  \hspace*{0.16in} {\bf Then} \\
\STATE \hspace*{0.3in} Set \textit{NashCondition=True}.
\STATE \hspace*{0.2in}{\bf End if}
\STATE {\bf End while}
\end{algorithmic}
\end{algorithm}

Therefore, by designing \eqref{predictor} as a PG according to Theorem \ref{t1}, the challenges caused by solution existence, algorithm convergence, and  computational scalability are  addressed. 

Moreover, the game \eqref{predictor} may have multiple PSNE, and Algorithm \ref{A_br} may converge to any of them, depending on the initial conditions. If using Algorithm \ref{A1}, then the outcome is not only a PSNE, but also a global minimizer of the potential function $F$. Since $F$ usually contains everyone's cost, e.g., \eqref{potential_design}, it represents a kind of ``social cost". A PSNE that is not only individually optimal for each player but also socially optimal from a global interest perspective, is naturally more preferable than others when multiple PSNE exist, addressing the challenge caused by the nonuniqueness of PSNE.   

\begin{algorithm}[t]
\caption{Potential function optimization to solve \eqref{predictor}} \label{A1}
\hspace*{0.0in} {\bf Input:} \\ 
\hspace*{0.2in}Agent set  $\mathcal{N}$;  \\
\hspace*{0.2in}System state  $x(t)$;  \\
\hspace*{0.2in}System dynamics  $f$;\\
\hspace*{0.2in}Strategy space  $\mathcal{A}$;  \\
\hspace*{0.2in}Cost functions  $\hat{J}_j$, $j\in\mathcal{N}$, that satisfy \eqref{cost_form};\\
\hspace*{0.02in} {\bf Output:} \\
\hspace*{0.2in} PSNE $\mathbf{a}^+(t)$.\\
\hspace*{0.02in} {\bf Procedures:} 
\begin{algorithmic}[1]
\STATE Find  $F(\mathbf{a}(t)) $ according to \eqref{potential_design}. 
\STATE Find $\mathbf{a}^+(t)$ such that
\begin{equation}\label{potential_optimize}
    \mathbf{a}^+(t)=\argmin_{\mathbf{a}(t)\in\mathcal{A}} F(\mathbf{a}(t)).
\end{equation}
\end{algorithmic}
\end{algorithm}

\subsection{Corrector}
 According to Definition \ref{d2}, the ego vehicle strategy generated from the Predictor, $\mathbf{a}_i^+(t)$,  is the best response to $\mathbf{a}_{-i}^+(t)$:
\begin{equation}\label{BR_predictor}
    \mathbf{a}_i^+(t)\in\argmin_{\mathbf{a}_i(t)\in\mathcal{A}_i} J_i(\mathbf{a}_i(t),\mathbf{a}_{-i}^+(t)).
\end{equation}
In other words, if  $\hat{J}_j=J_j$ $\forall j\in\mathcal{N}$, then $\mathbf{a}_i^+(t)$ is the ego vehicle's optimal strategy. 
However, {if $\hat{J}_j\neq J_j$, it is likely that}   $\mathbf{a}_{-i}^*(t)\neq \mathbf{a}_{-i}^+(t)$, where $\mathbf{a}_{-i}^*(t)=\{a_{-i}^*(t),a_{-i}^*(t+1),\cdots,a_{-i}^*(t+T-1)\}$ represents the surrounding agents' actual strategies.



Define action deviation at $t$ as 
\begin{equation}\label{actual_deviation}
    \omega_{-i}(t)=a_{-i}^*(t)-a_{-i}^+(t).
\end{equation}

In the Corrector step, we let the ego vehicle perform a best response to a corrected prediction on the surrounding agents' actions. Specifically, the ego vehicle aims to find $\mathbf{a}_i^*(t)$ such that
\begin{equation}\label{corrector}
    \mathbf{a}_i^*(t)\in\argmin_{\mathbf{a}_i(t)\in\mathcal{A}_i} J_i(\mathbf{a}_i(t),\mathbf{\hat{a}}_{-i}^*(t)),
\end{equation}
with
\begin{equation}\label{corrected_prediction}
\begin{split}
    \mathbf{\hat{a}}_{-i}^*(t)&=\mathbf{a}_{-i}^+(t)+\boldsymbol{\hat{\omega}}_{-i}(t),\\
    \boldsymbol{\hat{\omega}}_{-i}(t)&=\mathbf{1}_T\otimes\omega_{-i}(t-1),\\
   & =\mathbf{1}_T\otimes\left(a_{-i}^*(t-1)-a_{-i}^+(t-1)\right),
\end{split}
\end{equation}
where $\otimes$ represents Kronecker product, and $\mathbf{1}_T$ is a vector of ones with $T$ elements. {Note that $J_i$ in \eqref{corrector} is the cost function of the ego vehicle and hence can be assumed to be known.}

Equation \eqref{corrected_prediction} finds the corrected prediction on surrounding agents' actions, i.e., $\mathbf{\hat{a}}_{-i}^*(t)=\{\hat{a}_{-i}^*(t),\hat{a}_{-i}^*(t+1),\cdots,\hat{a}_{-i}^*(t+T-1)\}$,  taking into consideration both the prediction from the PG, i.e., $\mathbf{a}_{-i}^+(t)$, and the observed action deviation at $t-1$, i.e., $\boldsymbol{\hat{\omega}}_{-i}(t)=\mathbf{1}_T\otimes\omega_{-i}(t-1)$.  The ego vehicle best response to this corrected prediction is $ \mathbf{a}_i^*(t)$, derived by  \eqref{corrector}. 

At time $t$, let $\tau\in[t,t+T-1]$, $\tau\in\mathbb{Z}_+$. 
Define prediction error at $\tau$ as
\begin{equation}\label{prediction_error}
\begin{split}
    & e(\tau)=a_{-i}^*(\tau)-\hat{a}^*_{-i}(\tau).
\end{split}
\end{equation}
{The prediction error is of dimension $m_{-i}$, i.e., $e(\tau)\in\mathbb{R}^{m_{-i}}$.}
 
 Next theorem shows that with \eqref{corrector} and \eqref{corrected_prediction}, the prediction error $e(\tau)$ {admits a bound}.

\begin{thm}\label{t2}[Bounded prediction error]
{Assume that  $a_{-i}^*:\mathbb{Z}_+\rightarrow \mathbb{R}^{m_{-i}}$ and $a_{-i}^+:\mathbb{Z}_+\rightarrow \mathbb{R}^{m_{-i}}$ are Lipschitz continuous functions with 
\begin{equation}\label{bound1}
    \|a_{-i}^*(\tau)-a_{-i}^*(t-1)\|\leq K_1\cdot  (\tau-t+1)\cdot\Delta t,
\end{equation}
\begin{equation}\label{bound2}
    \|a_{-i}^+(\tau)-a_{-i}^+(t-1)\|\leq K_2\cdot  (\tau-t+1)\cdot\Delta t, 
\end{equation}
where $K_1\in\mathbb{R}_+$ and $K_2\in\mathbb{R}_+$ are two constants, $\Delta t\in\mathbb{R}_{++}$ is the sampling time, and $\tau\in[t,t+T-1]$, $\tau\in\mathbb{Z}_+$.}
Then there exists a constant $C\in\mathbb{R}_+$ such that the inequality,
\begin{equation}\label{error_bound}
    \|e(\tau)\|\leq C\cdot  (\tau-t+1)\cdot\Delta t, 
\end{equation} 
holds $\forall \tau\in[t,t+T-1]$, $\tau\in\mathbb{Z}_+$.
\end{thm}
\begin{proof}
Substituting \eqref{actual_deviation} and \eqref{corrected_prediction} into \eqref{prediction_error}, and since $\hat{\omega}_{-i}(\tau)=\omega_{-i}(t-1)$, $\forall\tau\in[t,t+T-1]$, $\tau\in\mathbb{Z}_+$, we have
\begin{equation}\label{error_1}
    \begin{split}
     e(\tau)&=a_{-i}^*(\tau)-\hat{a}^*_{-i}(\tau)\\
    &=\left(a_{-i}^+(\tau)+\omega_{-i}(\tau)\right)-\left(a_{-i}^+(\tau)+\hat{\omega}_{-i}(\tau)\right)\\
   &=\omega_{-i}(\tau)-\hat{\omega}_{-i}(\tau)\\
    &=\omega_{-i}(\tau)-\omega_{-i}(t-1)\\
    &=\left(a_{-i}^*(\tau)-a_{-i}^+(\tau)\right)-\left(a_{-i}^*(t-1)-a_{-i}^+(t-1)\right)\\
   &=\left(a_{-i}^*(\tau)-a_{-i}^*(t-1)\right)-\left(a_{-i}^+(\tau)-a_{-i}^+(t-1)\right).
\end{split}
\end{equation}

{Substituting} \eqref{bound1} and \eqref{bound2} into \eqref{error_1}, we have
\begin{equation}
\begin{split}
    &\|e(\tau)\|\\
    &=\left\|\left(a_{-i}^*(\tau)-a_{-i}^*(t-1)\right)-\left(a_{-i}^+(\tau)-a_{-i}^+(t-1)\right)\right\|\\
    &\leq\left\|a_{-i}^*(\tau)-a_{-i}^*(t-1)\right\|+\left\|a_{-i}^+(\tau)-a_{-i}^+(t-1)\right\|\\
    &\leq (K_1+K_2)\cdot(\tau-t+1)\cdot\Delta t.
    \end{split}
\end{equation}
By letting $C=K_1+K_2$, the proof is completed.
\end{proof}
Theorem \ref{t2} shows that the prediction error $e(\tau)$ remains bounded {over} a finite prediction horizon, and the bound depends on the constants $K_1$ and $K_2$.
The value of $K_1$ can be estimated from acceleration, jerk, and angular acceleration limits of the individual vehicles. For example, if the actions are the longitudinal acceleration/deceleration, then $K_1$ can be estimated from the vehicle jerk limit. According to \cite{jerk_limit}, $99.9\%$ vehicle jerk in the studied highway driving is within the range $[-1.47,1.07]m/s^3$. Therefore, one may select $K_1=1.47$ for highway scenarios (or select $K_1=2.94$ per the six sigma rules \cite{six_sigma} to enhance safety). In addition, since it is not likely that the peak jerk lasts for a long period, one may make $K_1$ smaller when $(\tau-t+1)\cdot\Delta t\geq 1s$, or limit the maximum acceleration to further narrow the error bounds. The value of  $K_2$ solely depends on the PG design, and thus can be manually selected (e.g., by letting $K_2=K_1$ to make it consistent with real trajectory bound) and be enforced when solving the PG. Additionally, it can also be estimated from offline closed-loop simulations of the multi-agent system operating according to the PG in Algorithm \ref{A1}.



Given the bound \eqref{error_bound}, we denote 
\begin{equation}\label{errorset}
    \mathcal{E}(\tau)=\{\hat{e}\in\mathbb{R}^{{m_{-i}}}|\|\hat{e}\|\leq C\cdot  (t+\tau-1)\cdot\Delta t\}
\end{equation} 
as the set of all possible prediction errors at $\tau$, $\tau\in[t,t+T-1]$.

The  Corrector {algorithm} is summarized in Algorithm \ref{A2}.

\begin{algorithm}[t]
\caption{PG Corrector} \label{A2}
\hspace*{0.0in} {\bf Input:} \\ 
\hspace*{0.2in}Output from the Predictor,  $\mathbf{\hat{a}}_{-i}^+(t)$;  \\
\hspace*{0.2in}Surrounding agents' actions at $t-1$, $a^*_{-i}(t-1)$;  \\
\hspace*{0.2in}Ego vehicle strategy space,  $\mathcal{A}_i$;  \\
\hspace*{0.2in}Ego vehicle cost function,  $J_i$;\\
\hspace*{0.02in} {\bf Output:} \\
\hspace*{0.2in} Optimal strategy for the ego vehicle,  $\mathbf{a}_i^*(t)$.\\
\hspace*{0.02in} {\bf Procedures:} 
\begin{algorithmic}[1]
\STATE \hspace*{0.02in}Find the corrected prediction $\mathbf{\hat{a}}_{-i}^*(t)$ according to \eqref{corrected_prediction}.
\STATE \hspace*{0.02in}Find the ego vehicle best response $\mathbf{a}_i^*(t)$ according to \eqref{corrector}.
\end{algorithmic}
\end{algorithm}

\begin{remark}
The ego vehicle cost functions in the Predictor and in the Corrector {do not need to be} identical, although in the current PCPG they are designed to be the same {to simplify the developments}. In the Predictor, the ego vehicle cost function needs to follow the form {required by} Theorem \ref{t1} {for} a PG. However, in the Corrector, it is not necessary, and the cost function can be designed in a more flexible way. 
\end{remark}

\begin{remark}
{Many} of the existing game-theoretic  AV decision-making {approaches} are open-loop \cite{suzhou,openloop1,Victor,my_potential}: Even if the ego vehicle observes other agents' action deviations, it {does not} adjust its own decision-making. A few works, on the other hand, aim to address this issue by learning driver-specific cost functions in real time \cite{closedloop1,dimitar_dynamics_3,online_payoff}. However, considering the short duration of agent interaction, {estimating others' cost functions can be challenging. In contrast, our approach relies on correcting the actions based on other agents' deviations and on exploiting the error bound \eqref{error_bound} predetermined offline.}
\end{remark}

\section{Safety and performance analysis}\label{performance}
This section analyzes the PCPG performance, including safety and optimality.

We define {a} safe set $\mathcal{X}_{i}^{\text{safe}}(x_{-i}(t))$ as the set of {all} $x_i(t)$ such that given $x_{-i}(t)$, if the ego vehicle state is within the safe set, i.e., $x_i(t)\in\mathcal{X}_{i}^{\text{safe}}(x_{-i}(t))$, then the ego vehicle is considered safe at $t$. 
An example of such a safe set, if $x$ represents vehicle  position, is $\mathcal{X}_{i}^{\text{safe}}(x_{-i}(t))=\{x_i(t)|\|x_i(t)-x_j(t)\|\geq d_{\text{safe}},\forall j\in\mathcal{N}_{-i}\}$, where $d_{\text{safe}}>0$ is a predefined safe distance. 

At time $t$, we denote the surrounding {agents' future state trajectories} generated by their actual action sequence $\mathbf{a}^*_{-i}(t)=\{a^*_{-i}(t),\cdots,a^*_{-i}(t+T-1)\}$ as $\{x^*_{-i}(t+1),\cdots,x^*_{-i}(t+T)\}$, unknown to the ego vehicle. For $\tau\in [t+1,t+T]$, denote by  $\hat{\mathcal{X}}_{-i}(\tau|t)$  the set of $\hat{x}_{-i}(\tau|t)$ generated by the set of action sequences $\{\hat{a}^*_{-i}(t)+\hat{e}(t),\cdots,\hat{a}^*_{-i}(\tau-1)+\hat{e}(\tau-1)\}$, where $\hat{e}(k)\in\mathcal{E}(k), k\in[t,\tau-1]$, and $\mathcal{E}(k)$ is defined in \eqref{errorset}. With a slight abuse of notation, $\hat{\mathcal{X}}_{-i}(\tau|t)$ can be represented as 
\begin{equation}\label{possibleactions}
\begin{split}
&\hat{\mathcal{X}}_{-i}(\tau|t)=\{\hat{x}_{-i}(\tau|t)|\hat{x}_{-i}(\tau|t)= \\
& f_{-i}\left(x_{-i}(t),\{\hat{a}^*_{-i}(t)+\hat{e}(t),\cdots,\hat{a}^*_{-i}(\tau-1)+\hat{e}(\tau-1)\}\right)\\
&|\hat{e}(k)\in\mathcal{E}(k), k\in[t,\tau-1]\},
\end{split}
\end{equation}
where 

$f_{-i}\left(x_{-i}(t),\{\hat{a}^*_{-i}(t)+\hat{e}(t),\cdots,\hat{a}^*_{-i}(\tau-1)+\hat{e}(\tau-1)\}\right)$ represents the surrounding agents' state at $\tau$ if the action sequence $\{\hat{a}^*_{-i}(t)+\hat{e}(t),\cdots,\hat{a}^*_{-i}(\tau-1)+\hat{e}(\tau-1)\}$ is implemented. 

Let $\mathcal{A}_{i}^{safe}(t,T)\subseteq\mathcal{A}_i$ denote the set of $\mathbf{a}_i(t)$ such that $\forall\tau\in [t+1,t+T]$,
\begin{equation}\label{25}
\begin{split}
    x_i(\tau)&=f_i\left(x_i(t),\{a_i(t),\cdots,a_i(\tau-1)\}\right)\\
    &\in\mathcal{X}_{i}^{safe}(\hat{x}_{-i}(\tau|t))
\end{split}
\end{equation}
holds $\forall \hat{x}_{-i}(\tau|t)\in\hat{\mathcal{X}}_{-i}(\tau|t)$.

Our next theorem shows that the outcome from the PCPG guarantees the ego vehicle safety, under suitable  conditions.

\begin{thm}\label{t3}[Safety]
If the ego vehicle cost function $J_i$ is designed such that
\begin{equation}\label{26}
    \argmin_{\mathbf{a}_i(t)\in\mathcal{A}_i} J_i(\mathbf{a}_i(t),\mathbf{\hat{a}}_{-i}^*(t))\subseteq \mathcal{A}_{i}^{safe}(t,T),
\end{equation}
and
\begin{equation}\label{27}
    \mathcal{A}_{i}^{safe}(t,T)\neq\emptyset,
\end{equation}
then $\mathbf{a}_i^*(t)$ guarantees the ego vehicle safety within the horizon $[t+1,t+T]$, i.e., $x_i(\tau)\in\mathcal{X}_{i}^{safe}(x^*_{-i}(\tau)), \forall\tau\in [t+1,t+T]$.
\end{thm}
\begin{proof}
According to Theorem \ref{t2},  $\forall\tau\in [t+1,t+T]$,
\begin{equation}
    a_{-i}^*(\tau-1)\in\{\hat{a}_{-i}^*(\tau-1)+\hat{e}(\tau-1)\text{ }|\text{ }\hat{e}(\tau-1)\in\mathcal{E}(\tau-1)\}.
\end{equation}
Therefore, according to \eqref{possibleactions}, 
\begin{equation}\label{29}
    x^*_{-i}(\tau)\in\hat{\mathcal{X}}_{-i}(\tau|t).
\end{equation}
Combining \eqref{25}, \eqref{26}, \eqref{27} and \eqref{29}, we have
\begin{equation}
    x_i(\tau)
    \in\mathcal{X}_{i}^{safe}(x^*_{-i}(\tau)).
\end{equation}
Therefore, the ego vehicle safety is maintained within the horizon $[t+1,t+T]$.
\end{proof}

\begin{remark}
Theorem \ref{t3} states that the ego vehicle safety is guaranteed with the PCPG, if the following two conditions are satisfied: 1) The ego vehicle is safety-conscious, i.e., \eqref{26} holds, and 2) A safe strategy exists, i.e., \eqref{27} holds.

A strategy $\mathbf{a}_i(t)$ is said to be safe if it satisfies \eqref{25}. Intuitively, it means that if a strategy leads to the ego vehicle safety against a set of surrounding agents' strategies, given by the prediction from the PCPG, $\hat{\mathbf{a}}_{-i}^*(t)$, and the bounded prediction error, $e(\tau)\in\mathcal{E}(\tau)$, then it is considered to be safe.

The condition in \eqref{26} represents a safety-conscious ego vehicle. That is, if safe strategies exist,  then a global minimizer of $J_i$ should be one of them. Such a cost function can be designed by incorporating the safety constraint as a barrier in $J_i$, using, for example, the interior-point method \cite{IPM1,IPM2}. A detailed example of how to design such a cost function is provided in Section \ref{simulation}. 
\end{remark}

The safety guarantee in Theorem \ref{t3} does not assume any specific behaviors of the surrounding agents. The surrounding agents can behave freely, be non-cooperative, or even be safety-agnostic. In the worst case, where all other agents aim to crash into the ego vehicle, the ego vehicle safety may not be maintained because a safe strategy does not exist. However, as long as a safe strategy exists, i.e., $\mathcal{A}_{i}^{safe}(t,T)\neq\emptyset$, the PCPG leads to the ego vehicle safety. Here $\mathcal{A}_{i}^{safe}(t,T)$ is defined based on the forward reachable set $\mathcal{X}_{i}^{safe}(\hat{x}_{-i}(\tau|t))$ in \eqref{25}, the computation of which can be realized using, for example, reachability-based trajectory design (RTD) \cite{RTD}.

In a situation where a safe strategy does not exist, i.e., $\mathcal{A}_{i}^{safe}(t,T)=\emptyset$, the outcome from the PCPG provides a minimum loss strategy for the ego vehicle, in the sense that $J_i$ is minimized subject to $\mathbf{\hat{a}}^*_{-i}(t)$. Note that $\mathcal{A}_{i}^{safe}(t,T)=\emptyset$ does not mean that a collision would definitely happen, as whether a collision happens depends on other agents' actions as well. If other agents are also safety-conscious, then the PCPG enables the ego vehicle to collaborate with others to avoid collision, because the outcome from the PG is not only individually optimal, but also  optimizes the ``team interest" of all agents, none of which desires a collision.

\begin{remark}
Theorem \ref{t3} does not guarantee recursive feasibility, unless other agents also satisfy certain safety constraints. That is, $\mathcal{A}_{i}^{safe}(t+1,T)$ is not guaranteed to be nonempty because we cannot control the surrounding agents' behaviors. As shown in \cite{RSS}, when AVs and human-driven vehicles share the road, it is impossible to guarantee absolute safety for an AV. On the other hand, if other agents' behaviors do satisfy certain safety constraints, e.g., the robust CBF (RCBF) constraints (i.e., Equation (2) in \cite{pcca}), then Theorem \ref{t3} can guarantee recursive feasibility, by letting $J_i$ incorporate the RCBF constraints as barriers. However, as we shall see in Section \ref{simulation}, if other agents are not safety-conscious, the CBF based approach can be less preferable than the PCPG. 
\end{remark}

Our next theorem studies the optimality of $\mathbf{a}_i^*(t)$. According to \eqref{corrector},  $\mathbf{a}_i^*(t)$ is optimal  if $\mathbf{\hat{a}}_{-i}^*(t)=\mathbf{a}_{-i}^*(t)$. With a slight abuse of notation, we denote $\left(\mathbf{a}_i^*(t),\mathbf{a}_{-i}^*(t)\right)$ (resp., $\left(\mathbf{a}_i^*(t),\mathbf{\hat{a}}_{-i}^*(t)\right)$) as the strategy profile that the surrounding agents take $\mathbf{a}_{-i}^*(t)$ (resp., $\mathbf{\hat{a}}_{-i}^*(t)$) and the ego vehicle take $\mathbf{a}_i^*(t)\in\argmin_{\mathbf{a}_i(t)\in\mathcal{A}_i} J_i(\mathbf{a}_i(t),\mathbf{a}_{-i}^*(t))$ (resp., $\mathbf{a}_i^*(t)\in\argmin_{\mathbf{a}_i(t)\in\mathcal{A}_i} J_i(\mathbf{a}_i(t),\mathbf{\hat{a}}_{-i}^*(t))$). 

\begin{thm}\label{t4}[Optimality]
Consider the PCPG designed in Section \ref{PCPG} and the action deviation in \eqref{actual_deviation}. If $\omega_{-i}(t)$ varies slowly with time, i.e., $\omega_{-i}(t)-\omega_{-i}(t-1)\rightarrow \mathbf{0}$, $\forall t\in \mathbb{Z}_+$, where $\mathbf{0}$ is a vector of zeros of proper dimensions, then $\left(\mathbf{a}_i^*(t),\mathbf{\hat{a}}_{-i}^*(t)\right)\rightarrow\left(\mathbf{a}_i^*(t),\mathbf{a}_{-i}^*(t)\right)$. 
\end{thm}

\begin{proof}
If $\omega_{-i}(t)-\omega_{-i}(t-1)\rightarrow \mathbf{0}$ and $T$ is finite, then according to \eqref{error_1},  
\begin{equation}\label{T4_eq}
\begin{split}
    &a_{-i}^*(\tau)-\hat{a}^*_{-i}(\tau)\\
    &\quad=\omega_{-i}(\tau)-\omega_{-i}(t-1)\\
    &\quad\rightarrow \mathbf{0}, \quad \forall \tau\in[t,t+T-1], \forall t\in \mathbb{Z}_+.
    \end{split}
\end{equation}
Since from \eqref{T4_eq}, $\hat{a}_{-i}^*(\tau)\rightarrow a^*_{-i}(\tau)$ holds for all $\tau\in[t,t+T-1]$ and $\forall t\in \mathbb{Z}_+$, we have $\mathbf{\hat{a}}_{-i}^*(t)\rightarrow\mathbf{a}_{-i}^*(t)$. Therefore,   $\left(\mathbf{a}_i^*(t),\mathbf{\hat{a}}_{-i}^*(t)\right)\rightarrow\left(\mathbf{a}_i^*(t),\mathbf{a}_{-i}^*(t)\right)$.
\end{proof}
\begin{remark}
Theorem \ref{t4} suggests that if the action deviation $\omega(t)$ varies slowly with time, then the outcome from the PCPG accurately approximates the actual PSNE $\left(\mathbf{a}_i^*(t),\mathbf{a}_{-i}^*(t)\right)$. 
In practice, a slowly time-varying  $\omega(t)$ indicates a consistent driving style. 
For example, if a driver behaves aggressively at $t-1$, which may be reflected by a larger-than-typical acceleration, i.e., $a_j^*(t-1)-a_j^+(t-1)>0$, it is reasonable to assume that he/she would continue this aggressiveness for a while, i.e.,  $a_j^*(t+k)-a_j^+(t+k)>0$ for some $k=0,1,\cdots$. 

From a reasoning perspective, such a  slowly time-varying  $\omega(t)$ indicates a consistent reasoning behind agent behaviors. For example, if agent $j$ cares about safety less than a typical driver (i.e., the weight of the collision avoidance term in $J_j$ is smaller than in $\hat{J}_j$), which may cause a more moderate braking than expected, i.e., $0>a_j^*(t-1)>a_j^+(t-1)$, then it is reasonable to assume that this inequality also holds at $t$, i.e., $0>a_j^*(t)>a_j^+(t)$, since his/her reasoning, characterized by $J_j$,  remains the same at time $t$. 


\end{remark}

Note that although a slowly time-varying  $\omega(t)$ facilitates achieving near optimality, it is not required or necessary in the safety guarantee, i.e., Theorem \ref{t3}.

To summarize, with the PCPG, the ego vehicle plans its strategy according to $\hat{J}_j$ (Predictor), and gets the strategy improved by incorporating others' action deviations into its decision-making (Corrector). Intuitively, although others' actual cost functions are unknown, the ego vehicle can infer how different they are from the assumed ones, e.g., whether the aggressiveness is underestimated or overestimated, by observing  $a_j^*(t-1)-a_j^+(t-1)$. With this information, the prediction error on others' actions can be effectively compensated without requiring accurate driver-specific cost functions.

 \begin{remark}
Note that the agent set $\mathcal{N}$ can be a subset of the traffic agents, depending on the application needs and the road structure. For example, in the 3-lane highway scenario studied in \cite{Ford_1}, up to 6 surrounding vehicles are considered as interactive agents. In the merging scenario studied in \cite{kaiwen}, 3 surrounding vehicles are selected according to a headway-based rule. These agent selection rules are also applicable in the proposed PCPG setting.  The performance of the PCPG with a subset of traffic agents is reported in Section \ref{simulation_validation}.
\end{remark}

\section{Numerical studies}\label{simulation}

This section applies  the  PCPG framework to specific traffic scenarios. To illustrate the scalability of the PCPG to varied environments/traffic scenarios, we test the PCPG performance in three different scenarios: two-vehicle oncoming traffic scenario, five-vehicle intersection-crossing scenario, and multi-vehicle highway merging scenario. In addition, validation results using a naturalistic driving dataset are also reported in highway merging scenario.

The vehicles' dynamics are described by the following equations \cite{dynamics}.
\begin{equation}
\begin{split}
    &X_i(t+1)=X_i(t)+  v_{i}(t)\Delta t ,\\
    &v_{i}(t+1)=v_{i}(t)+ a_{i}(t)\Delta t,
    \end{split}
\end{equation}
where $i=1,2,\cdots,N$ are the traffic agents, $X_i=[x_i,y_i]^T$ represents agent $i$'s position, $v_i=[v_{x,i},v_{y,i}]^T$ is the speed, and  $a_i=[a_{x,i},a_{y,i}]^T$ is the  acceleration, and here, agent $i$'s action. Let the sampling time $\Delta t$ be $0.5s$ in the simulation.

Agents' cost functions  are designed as
\begin{equation}\label{cost_predictor}
\begin{split}
    J_i(\mathbf{a}_i(t),\mathbf{a}_{-i}(t))=\theta_iJ_i^{self}(\mathbf{a}_i(t))+\sum_{j\in \mathcal{N},j\neq i}J_{ij}(\mathbf{a}_i(t),\mathbf{a}_j(t)),
\end{split}
\end{equation}
where $\mathbf{a}_i(t)=\{a_i(t),a_i(t+1),\dots,a_i(t+T-1)\}$ is agent $i$'s strategy generated at $t$ over the prediction horizon of length $T$, and $\theta_i$ is a constant that characterizes agent $i$'s aggressiveness. 

The first term in \eqref{cost_predictor} is designed to track the desired position and speed. Specifically,
\begin{equation}\label{cost_self}
\begin{split}
  &J_i^{self}(\mathbf{a}_i(t))\\
  &= \sum_{\tau=t}^{t+T-1} \left(\Delta X_i(\tau)^TQ_i\Delta X_i(\tau)+\Delta v_i(\tau)^TR_i\Delta v_i(\tau)\right),
\end{split}
\end{equation}
where $\Delta X_i(\tau)=X_i(\tau)-X_{i}^d$, $\Delta v_i(\tau)=v_i(\tau)-v_{i}^d$, $X_{i}^d$ and $v_{i}^d$ are the desired position and speed, respectively. The desired position usually represents, for example, staying in the center of a lane. Here $Q_i$ and $R_i$ are weighting matrices,  $Q_i$ is positive semi-definite, and $R_i$ is positive definite. 

The second term in \eqref{cost_predictor} is to avoid collision and is designed to be
\begin{equation}\label{cost_collision}
 J_{ij}(\mathbf{a}_i(t),\mathbf{a}_j(t))=\sum_{\tau=t}^{t+T-1} \frac{d_d^2}{d_{ij}^2(\tau)+\delta},   
\end{equation}
 where $d_d$ is the  comfortable (or desired)  inter-vehicle distance, $\delta> 0$ is a positive small number to avoid the denominator being $0$. The cost \eqref{cost_collision} discourages the actions that lead to small inter-vehicle distance $d_{ij}(\tau)$:
 \begin{equation}\label{distance}
     d_{ij}(\tau)=\sqrt{\left(x_i(\tau)-x_j(\tau)\right)^2+\left(y_i(\tau)-y_j(\tau)\right)^2}.
 \end{equation}
As shown from \eqref{cost_predictor}-\eqref{distance}, agent $i$'s cost is affected by not only its own states and actions but also the states and actions of its surrounding agents. In light traffic, where $d_{ij}$ is sufficiently large for all $j$, agent $i$ may be able to track its desired speed $v_i^d$, since its behavior is primarily governed by \eqref{cost_self} in this case. However, in dense traffic, agent $i$ may not be able to maintain its desired speed, since the safety-related cost \eqref{cost_collision}-\eqref{distance} would increase dramatically with the decrease of inter-vehicle distances.

In our simulation, the strategy space $\mathcal{A}_i$ is selected such that a constant action is planned over the prediction horizon, i.e.,  $a_i(\tau)=a_i(t)$, $\forall \tau\in[t,t+T-1]$. Note that although the AV {plans} one maneuver over the horizon, it may change its mind and select another maneuver after $\Delta t$, triggered by the receding horizon control.  This setting is consistent with a common driving experience: When a human driver plans a maneuver, e.g., steering to change lane,  he/she usually expects to continue this maneuver for some time, e.g., $4s$ for lane-changing \cite{add_T}.

\subsection{Oncoming traffic scenario}\label{oncoming}
This subsection considers the oncoming traffic scenario, where two vehicles encounter each other on a narrow road, as shown in Figure \ref{f_oncoming}. In such a scenario, both vehicles desire to keep their lateral positions, maintain the lateral and longitudinal velocities, and avoid collisions. To safely pass each other, vehicles have to compromise and deviate their trajectories and velocities from the individually desired ones when they cross each other. How much a vehicle would compromise depends on its aggressiveness, i.e., $\theta_i$. This scenario is challenging to handle even for human drivers, as it requires careful interaction with the other driver.  Wrong prediction on the other's intention/strategy can easily lead to a collision. We consider limited control authority for both vehicles, i.e., $a_{x,i}(t)\in[-3,3]$ $m/s^2$ and $a_{y,i}(t)\in[-3,3]$ $m/s^2$, $i=1,2$. 
\begin{figure}[thpb]
\centering
\includegraphics[width=0.18\textwidth]{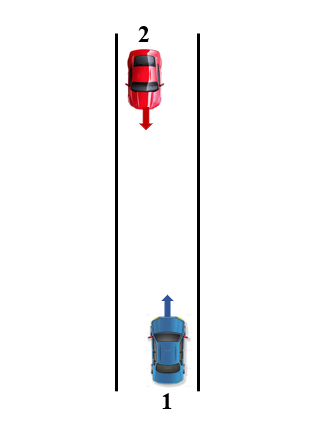}
\caption{Oncoming traffic scenario: two vehicles encounter each other on a narrow road.}\label{f_oncoming}
\end{figure}
\textbf{Study 1: PG alone.} We first test the performance of the Predictor alone, i.e., the potential game with predefined cost functions. The game is solved using the potential function optimization, i.e., Algorithm \ref{A1}. 

We first let $\hat{J}_2=J_2$, that is, vehicle 2's cost function is known to the ego vehicle. In this situation, the two vehicles successfully bypass each other without collision, as shown in Figure \ref{PG_smart}. (The animation is available in Youtube at \textit{https://www.youtube.com/watch?v=UOefR0Dhhqk}.) The dotted circles represent the safe distance, and have the radius of $2m$ for each vehicle. If the circles cross, a collision is considered to happen. 

We then consider the case $\hat{J}_2\neq J_2$. Specifically, the ego vehicle assumes $\theta_2=1$  (corresponding to $\hat{J}_2$), however, $\theta_2=10$ (corresponding to $J_2)$, representing a more aggressive vehicle 2 than the ego vehicle expects. In this situation, the ego vehicle expects cautious behavior from vehicle 2, i.e., deviating its desired trajectory and speed sufficiently to maintain safety. However, vehicle $2$ does not compromise so much, leading to collision at $t=6.2s$, as shown in Figure \ref{PG_stupid}.


\begin{figure}
    \centering
    \includegraphics[width=0.3\textwidth]{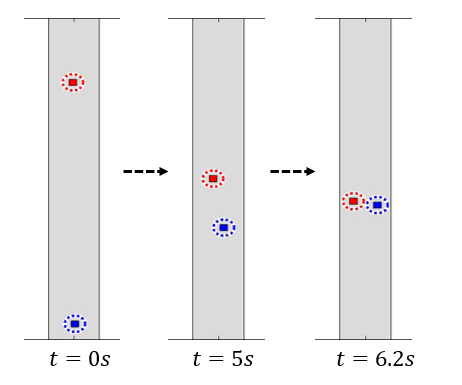}
    \caption{Oncoming traffic scenario with PG alone. $\hat{J}_2=J_2$. Two vehicles collaboratively bypass each other. (Video is available at \textit{https://www.youtube.com/watch?v=UOefR0Dhhqk}.) }
    \label{PG_smart}
\end{figure}

\begin{figure}
    \centering
    \includegraphics[width=0.3\textwidth]{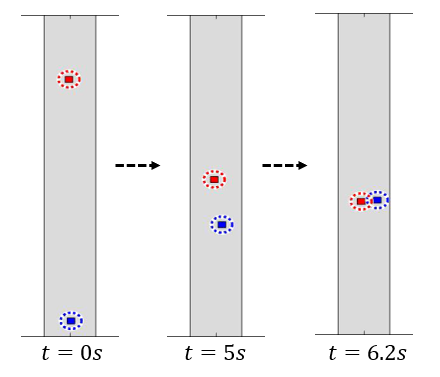}
    \caption{Oncoming traffic scenario with PG alone.  $\hat{J}_2\neq J_2$. A collision happens at $t=6.2s$.}
    \label{PG_stupid}
\end{figure}

\textbf{Study 2: PCPG.} With the same setting and initial conditions as in Figure \ref{PG_stupid}, we then test the performance of the PCPG. The result is shown in Figure \ref{oncoming_PCPG}. Thanks to the Corrector, the ego vehicle notices that vehicle $2$ does not behave as expected, and therefore, it adjusts its decision-making by feeding back vehicle 2's actual behaviors according to Algorithm \ref{A2}. The collision is now successfully avoided, despite the mis-information $\hat{J}_2\neq J_2$.

\begin{figure}
    \centering
    \includegraphics[width=0.30\textwidth]{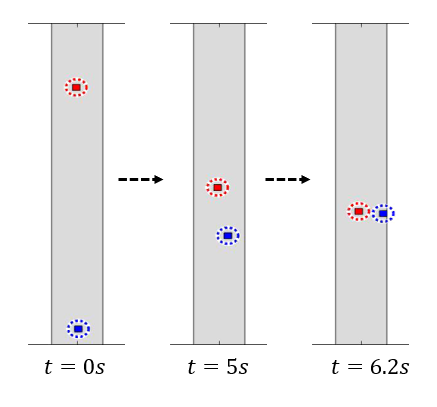}
    \caption{Oncoming traffic scenario with the PCPG. Collision is successfully avoided despite $\hat{J}_2\neq J_2$. }
    \label{oncoming_PCPG}
\end{figure}

To better visualize the two vehicles' behaviors and the ego vehicle's expectations, we plot their trajectories in Figure \ref{oncoming_trajectory}. 
The blue and red circles represent the trajectories of vehicles 1 and 2, i.e., $X_1(t)$ and $X_2(t)$, respectively. The ego vehicle trajectory is from the PCPG. The yellow lines represent the ego vehicle predicted vehicle 2 trajectory in the PG,  at each $t$ with a horizon $T=2s$. The green lines represent the corrected prediction in the PCPG. It is clear that the prediction error is significantly reduced in the PCPG compared to  the PG, as the green lines, compared to the yellow ones, are much closer to the red. With the corrected prediction error, the ego vehicle can better respond to the approaching of vehicle 2. 

\begin{figure}
    \centering
    \includegraphics[width=0.4\textwidth]{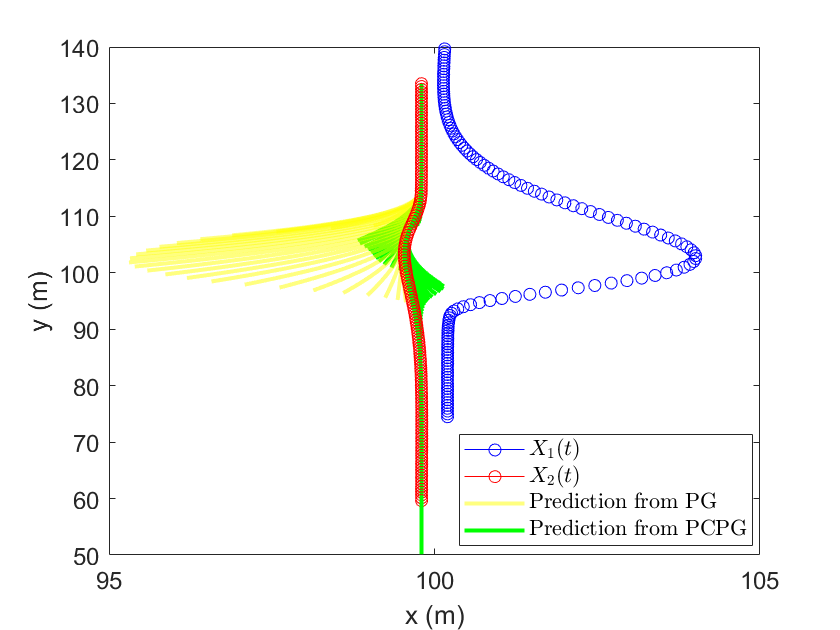}
    \caption{Trajectories of vehicles 1 and 2. Blue circles: $X_1(t)$, Red circles: $X_2(t)$, Yellow lines: predicted $X_2(\tau)$ according to $\mathbf{a}_2^+$ (PG), and Green lines: predicted $X_2(\tau)$ according to $\mathbf{\hat{a}}_2^*$ (PCPG), where $\tau\in[t,t+2](s)$. }
    \label{oncoming_trajectory}
\end{figure}

\textbf{Study 3: PCPG vs. CBF.} In addition, we also compare the PCPG with a CBF based quadratic program (QP) approach. As  vehicles are controlled in a decentralized manner,  the centralized CBF approach of \cite{cbf} is not applicable. As such, we employ a variant of the centralized CBF -- the predictor corrector collision avoidance (PCCA) of  \cite{pcca}. In the PCCA, the ego vehicle solves the following constrained optimization at each $t$.
\begin{equation}\label{PCCA}
\begin{split}
    &\min_{a_1(t),a_2(t)}\left(\|a_1(t)-a_{10}(t)\|^2+\|a_2(t)\|^2\right)\\
    s.t.& \quad b(t)+c(t)a_1(t)-c(t)a_2(t)-c(t)\Tilde{\omega}_2(t)\geq 0
\end{split}
\end{equation}
where $\Tilde{\omega}_2(t)=a_2^*(t-1)-\Tilde{a}_2(t-1)$, $\Tilde{a}_2(t-1)$ is vehicle 2 action from the PCCA at $t-1$, and $a_{10}(t)$ is a baseline controller and is from  a linear quadratic regulator (LQR). The constraint in \eqref{PCCA} is from the robust CBF condition, with a relative-degree-two CBF. Specifically,
\begin{small}
\begin{equation}
\begin{split}
    &b(t)\\
    &=2v_{12}^T(t)v_{12}(t)+2l_1X_{12}(t)^Tv_{12}(t)+l_0(X_{12}^T(t)X_{12}(t)-d_{safe}^2),\\
   &c(t)=2X_{12}^T(t),
    \end{split}
\end{equation}
\end{small}
where $v_{12}(t)=v_{1}(t)-v_{2}(t)$, $X_{12}(t)=X_{1}(t)-X_{2}(t)$, and $l_1$ and $l_0$ are two constants and are selected to be the same as in \cite{pcca}.  Vehicle 2 behaviors remain the same as in Study 2. 

\begin{figure}
    \centering
    \includegraphics[width=0.45\textwidth]{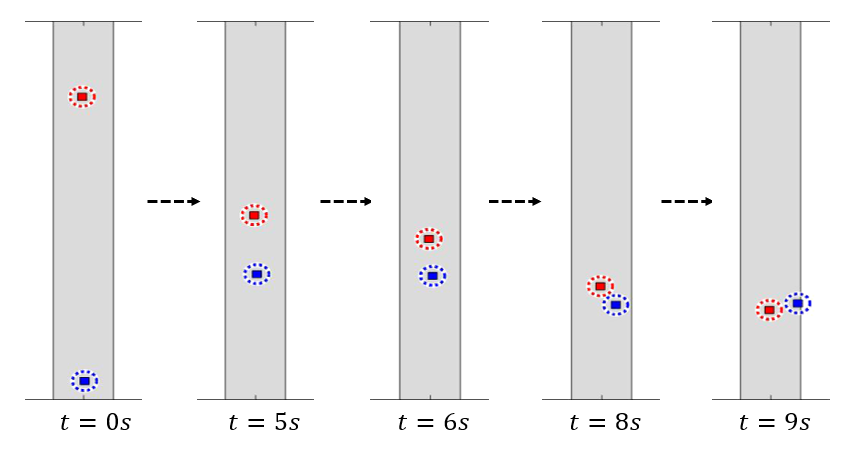}
    \caption{Oncoming traffic scenario with the PCCA. }
    \label{oncoming_pcca}
\end{figure}

As the original PCCA does not consider the control saturation, we first test the  PCCA with unlimited control authority, i.e., unbounded action space $a_1(t)\in\mathcal{U}_1=\mathbb{R}^2$. All other settings and initial conditions are selected to be the same as in Study 2. The result is shown in Figure \ref{oncoming_pcca}. As expected, the PCCA  successfully avoids collision. Notably, the ego vehicle performs  different behaviors with the PCPG and with the PCCA, by comparing Figures \ref{oncoming_PCPG} and \ref{oncoming_pcca}. Specifically, in the PCCA, the ego vehicle responds at a much later time compared to in the PCPG (e.g., at $t=5s$, the ego vehicle is well-prepared for the bypassing in the PCPG, while it does not respond yet in the PCCA, since the constraint in \eqref{PCCA} is not active yet). Because of this later response, a much larger $\|a_1(t)\|$ is triggered after the constraint in \eqref{PCCA} becomes active   (e.g., the ego vehicle in Figure \ref{oncoming_pcca} moves backward during $6s-8s$ to avoid collision). To quantify this observation and the AV driving performance in addition to safety, let us consider the following two performance metrics:

\textit{a) Deviation from the desired speed}. We calculate both the average and the maximum longitudinal speed deviations, which are related to travel efficiency, according to
\begin{equation}\label{travel_eff}
\begin{split}
    \Delta \bar{v}_{y,1}&=\frac{\sum_{t=0}^{t_{f}}\Delta v_{y,1}(t)}{t_{f}},\\
   \Delta v_{y,1}^{max}&=\max_{t\in[0,t_f]}\Delta v_{y,1}(t),
\end{split}
\end{equation}
where
\begin{equation}
       \Delta v_{y,1}(t)=(v_{y,1}(t)-v_{y,1}^d),
\end{equation}
$t_{f}=\frac{15s}{0.5s}=30$ is the total time steps of the simulation, and $v_{y,1}^d=5m/s$ is the ego vehicle desired longitudinal speed. Smaller $\Delta \bar{v}_{y,1}$ and $\Delta v_{y,1}^{max}$ indicate better adherence to the desired speed, and thus, are preferred.  

\textit{b) Deviation from the desired heading}. The average and the maximum heading angle deviations, which are related to ride comfort \cite{comfort}, are calculated by 
\begin{equation}
\begin{split}
    \Delta \bar{\phi}_{y,1}&=\frac{\sum_{t=0}^{t_{f}} \Delta \phi_{1}(t)}{t_{f}},\\
    \Delta \phi_{1}^{max}&=\max_{t\in[0,t_f]} \Delta\phi_{1}(t),
    \end{split}
\end{equation}
where
\begin{equation}\label{40}
\begin{split}
\Delta\phi_{1}(t)&=\phi_{1}(t)-\phi_{1}^d=\arctan\frac{v_{y,1}(t)}{v_{x,1}(t)}-\phi_{1}^d,
\end{split}
\end{equation}
 $\phi_{1}^d=90^{\circ}$ is the ego vehicle desired heading angle. Small $\Delta \bar{\phi}_{y,1}$ and $\Delta \phi_{1}^{max}$ are preferred.

With the above metrics, we compare the performance of the PCPG and the PCCA. The PCPG leads to 
\begin{equation}\nonumber
    \begin{split}
        \Delta \bar{v}_{y,1}&=0.22m/s, \quad \quad
        v_{y,1}^{max}=1.29m/s\\
        \Delta \bar{\phi}_{y,1}&=6.48^{\circ}, \quad \quad \quad
        \Delta \phi_{1}^{max}=36^{\circ}
    \end{split}
\end{equation}
The PCCA leads to
\begin{equation}\nonumber
    \begin{split}
        \Delta \bar{v}_{y,1}&=1.82m/s, \quad \quad
        v_{y,1}^{max}=8.63m/s\\
        \Delta \bar{\phi}_{y,1}&=36.68^{\circ}, \quad \quad \quad
        \Delta \phi_{1}^{max}=170.77^{\circ}
    \end{split}
\end{equation}
Thus the PCPG results in better adherence to the desired trajectory, thanks to its early response realized by the ego vehicle predictive capability on vehicle 2 strategies, enabling the ego vehicle to plan ahead and to take necessary and moderate maneuvers in advance. In contrast, in the PCCA, the ego vehicle simply reacts to the instantaneous states and actions of vehicle 2, and does not have a capability to predict vehicle 2 future behaviors or to plan ahead. 

Next let us consider limited control authority for the ego vehicle, i.e., $a_{x,1}(t)\in[-3,3]$ $m/s^2$ and $a_{y,1}(t)\in[-3,3]$ $m/s^2$. Specifically, if $a_{x,1}(t)$ or $a_{y,1}(t)$ from \eqref{PCCA} is larger than $3$ $m/s^2$ (resp., less than $-3$ $m/s^2$), we make the ego vehicle implement $3$ $m/s^2$ (resp., $-3$ $m/s^2$). In this case, due to the late response and the bounded action space, a collision happens if the ego vehicle employs the PCCA. The performance is similar to Figure \ref{PG_stupid}. 

 \begin{table*}[!h]
\centering
\caption{Statistical comparison}
\begin{tabular}{c|c|c|c|c|c|c}
\hline  
\multicolumn{1}{c|}{\multirow{2}{*}{}}
&\multicolumn{3}{c|}{\multirow{2}{*}{Two-vehicle oncoming traffic}}
&\multicolumn{3}{c}{\multirow{2}{*}{Five-vehicle intersection-crossing}}\\
\multicolumn{1}{c|}{\multirow{2}{*}{}}
&\multicolumn{3}{c|}{\multirow{2}{*}{}}
&\multicolumn{3}{c}{\multirow{2}{*}{}}
\\
\hline
Ego vehicle controller & PG  & PCPG & PCCA & PG  & PCPG & PCCA
\\
\hline
\multirow{1}{*}{Collision rate} & $264/500$ & $0/500$ & $37/500$ & $28/500$ & $0/500$ & $78/500$
\\
\hline
\multirow{1}{*}{Ave/Max longitudinal speed deviation (m/s)} & $0.11/1.14$  & $0.20/3.63$  & $0.90/8.30$  & $1.00/5.90$  & $1.56/8.20$ & $0.74/7.50 $
\\
\hline
\multirow{1}{*}{Ave/Max heading angle deviation} & $3.75^{\circ}/38.26^{\circ}$ & $6.50^{\circ}/69.70^{\circ}$ & $14.85^{\circ}/179.8^{\circ}$ & $0^{\circ}/0^{\circ}$ & $0^{\circ}/0^{\circ}$ & $0^{\circ}/0^{\circ}$
\\
\hline
\multirow{1}{*}{Ave/Max computational time (s)} & $0.06/0.25$   & $0.10/0.33$   & $<0.01/<0.01$   & $0.08/0.30$  & $0.11/0.34$   & $<0.01/<0.01$  
\\
\hline
\end{tabular}\label{table1}
\end{table*}

\textbf{Study 4: Statistical studies.}  To make the comparative results  more convincing, we conduct statistical studies. $500$ scenarios are tested with randomly selected initial lateral  positions of the two vehicles, i.e., $x_2(0)-x_1(0)$ is uniformly distributed in $[0.2,2]$ $m$, and vehicle 2's aggressiveness, i.e.,  $\theta_2$ is uniformly distributed in $[1,10]$. The ego vehicle always assumes  $\theta_2=1$ in all scenarios. Both vehicles have limited control authorities. The statistical results are shown in Table \ref{table1}, where ``collision rate'' represents the number of scenarios where collision happens divided by the total tested scenarios, ``speed deviation" and ``heading angle deviation" are calculated according to  \eqref{travel_eff}-\eqref{40}, and
``computational time" represents the running time for each decision-making collected from MATLAB$^{\circledR}$ on a laptop with an Intel Core i7-10750H processor clocked at $2.60$ GHz and $16$ GB of RAM. The optimization in the PG and in the PCPG approaches is performed using the Matlab genetic algorithm function `\textit{ga}' \cite{ga}, and in the PCCA approach using the quadratic programming function `\textit{quadprog}'\cite{qp}. 

As we can observe from Table \ref{table1}, among the three approaches, PG, PCPG, and PCCA, the PCPG is the only one that ensures the ego vehicle safety in all tested scenarios. Meanwhile, the PCPG also leads to reasonably small speed deviations and heading angle deviations. Moreover, all of the three approaches are computationally practical, as all their running time is significantly less than the sampling time 0.5s. The PCCA has the lightest computational load and its running time is always less than 0.01s. While the PCPG has the best performance in both collision avoidance and mobility, between PCCA and PG,  the PCCA is more effective in avoiding collisions in the tested scenarios. It is because safety is a hard constraint in the PCCA, while a soft constraint in the PG, and the effectiveness of the PG is heavily dependent on reliable prediction of the surrounding agents' behaviors.

\subsection{Multi-vehicle intersection-crossing}\label{intersection}
This subsection considers a multi-vehicle intersection-crossing scenario, as shown in Figure \ref{f_intersection}. 
 The vehicle labeled with the number ``1" is the ego  {vehicle}, aiming to go straight to cross the intersection. The heading directions of other vehicles are marked using grey arrows, and are known to the ego vehicle. Each vehicle controls its own longitudinal acceleration, and the lateral acceleration is  zero, i.e., no lane-changing while crossing an intersection. The action space for each vehicle is {$\mathcal{U}_i=[-3,3]$ $m/s^2$}.
 
 \begin{figure}[thpb]
\centering
\includegraphics[width=0.26\textwidth]{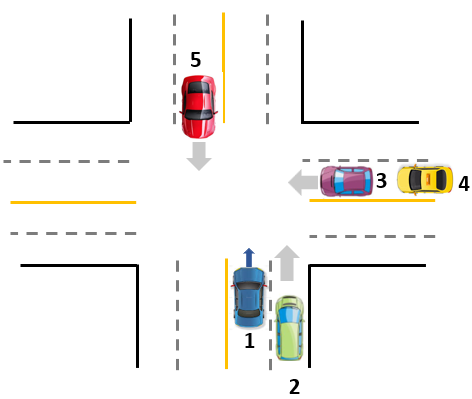}
\caption{Multi-vehicle intersection-crossing scenario}\label{f_intersection}
\end{figure}

The cost {function}  is designed according to \eqref{cost_predictor}-\eqref{cost_collision}, where $\theta_1=1$ and $\theta_j\in[1,100]$ for $j=2,3,4,5$, unknown to the ego vehicle. Note that when $\theta_j=100$, the collision avoidance term in \eqref{cost_predictor} is very lightly weighted, indicating a safety-agnostic vehicle.  To make the scenario more challenging for the ego vehicle, we also make the surrounding vehicle desired speed $v_{j}^d$ unknown to the ego vehicle. The ego vehicle assumes that $v_{j}^d=5m/s,  \forall j$, but $v_{j}^d\in[5,15]m/s$ for $j=2,3,4,5$. 

With the above setting, we test and compare the performance of the PG, PCPG, and PCCA. 

 \begin{figure*}[thpb]
\centering
\includegraphics[width=0.9\textwidth]{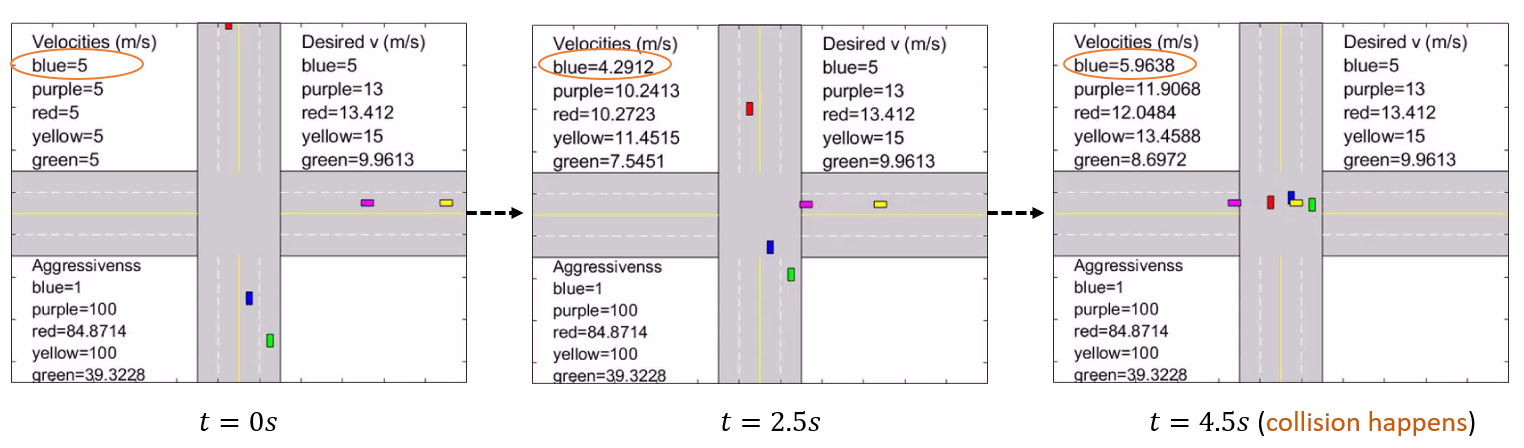}
\caption{Multi-vehicle intersection-crossing with PG alone. The blue is the ego vehicle. A collision happens due to the ego vehicle mis-information.}\label{intersection_PG}
\end{figure*}

 \begin{figure*}[thpb]
\centering
\includegraphics[width=0.9\textwidth]{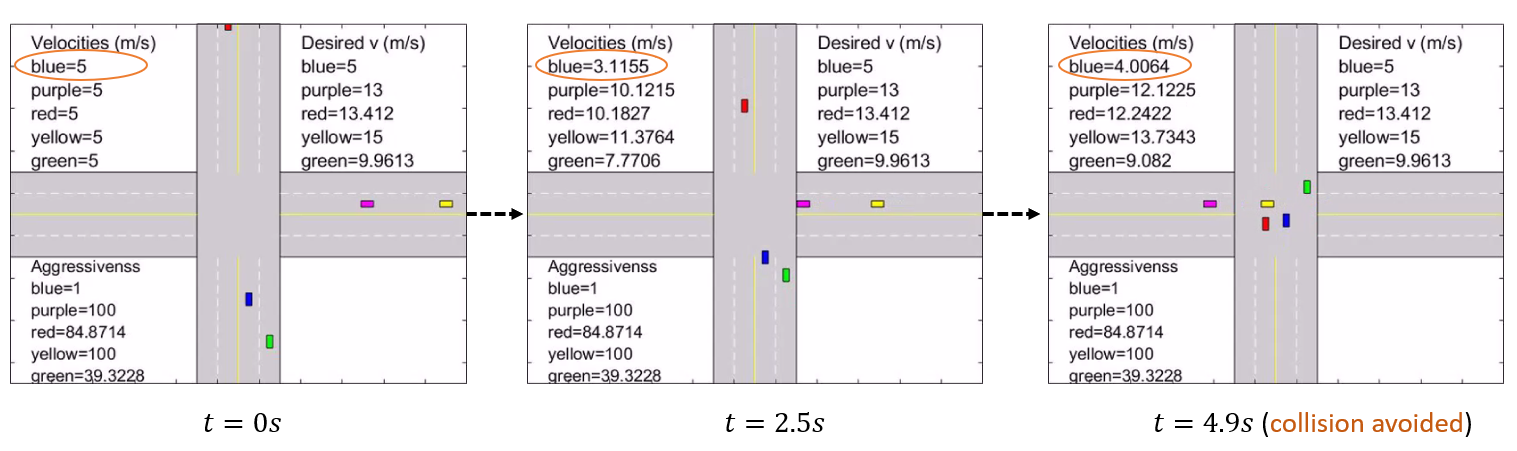}
\caption{Multi-vehicle intersection-crossing with PCPG. The collision is successfully avoided despite $\hat{J}_j\neq J_j$, $j\in\mathcal{N}_{-i}$. }\label{intersection_pcpg}
\end{figure*}

In the PG, the ego vehicle solves a $5$-player game with the assumed cost functions $\hat{J}_j$, $j=2,3,4,5$. When $\hat{J}_j$ and $J_j$ are close, the ego vehicle performs reasonably well in terms of both safety and mobility. However, if $J_j$ deviates significantly from $\hat{J}_j$, a collision may happen. Figure \ref{intersection_PG} shows one such scenario. Vehicles' aggressiveness, instantaneous and desired speeds are all labeled in the figure. In this pictured scenario, the purple and the yellow vehicles are both safety-agnostic ($\theta_j=100$) and have much larger desired speed than the ego vehicle expects ($v_{3}^d=13m/s$ and $v_{4}^d=15m/s$). At around $4s$, a collision with the yellow vehicle happens, as the ego vehicle expects the yellow to slow down, while the yellow speeds up to track its desired speed.

With the same setting as in Figure \ref{intersection_PG}, we test the PCPG. As shown in Figure \ref{intersection_pcpg}, the collision is now successfully avoided. The ego vehicle notices the action deviation of the purple and the yellow vehicles, and  improves its strategies accordingly. It can be seen by comparing $v_1(t)$ at $t=2.5s$  in the PCPG (Figure \ref{intersection_pcpg}) and in the PG (Figure  \ref{intersection_PG}). 

For the PCCA, similar to the oncoming traffic scenario, if the ego vehicle has unlimited control authority, then collisions can always be avoided. If the control authority is limited, then collisions may happen. Due to the page limit, we do not show the detailed scenario pictures, instead, we present the statistical results in Table \ref{table1}.

From Table \ref{table1}, the PCPG performs the best among the three approaches in terms of safety, as it is the only one that ensures no collision. 
Different from the two-vehicle scenarios, the PG in the five-vehicle scenarios performs better than the PCCA in terms of safety. It is because the PG encodes the global agent interactions in its decision-making, while the PCCA encodes agent interactions in a pairwise manner (i.e., the CBF constraints are for each vehicle pair separately \cite{pcca}). It is possible that a strategy solving the conflict with one vehicle  worsens the situation with another. Since the CBF constraints are not always active for all vehicles, it is possible that when the constraint becomes active with one vehicle, the ego vehicle is handling the conflict with another, and their handling strategies are contradicting each other, leading to collision with at least one of the two vehicles.  Therefore, the CBF-based approach can be ``short-sighted" in multi-vehicle scenarios, due to the inability to capture global agent interactions. 

From the computational perspective, all the three approaches are practical, since their running time is always less than the sampling time. Comparing the running time in the two-vehicle and in the five-vehicle scenarios, we can see that the  time cost does not increase much with the increase of the number of agents. Specifically, with the PCPG, the average computational time is $0.10s$ in the two-vehicle scenario and $0.11s$ in the five-vehicle scenario, indicating good computational scalability. 

{
\subsection{Validation with naturalistic traffic data}\label{simulation_validation}
This subsection validates the PCPG performance using naturalistic traffic dataset provided by Federal Highway Administration's (FHWA's) Next Generation Simulation (NGSIM) program \cite{I80}. The data was collected  on a segment of Interstate 80 (I-80) in Emeryville (San Fransico), California. A snapshot of the I-80 freeway and the traffic situation is shown in Figure\ref{I80}. We focus on the merging area (the same road segment and data extraction as in \cite{zhaojian}) and aim to validate whether the PCPG enables safe and efficient merging in dense traffic. To this end, we test the PCPG performance on $10$ merging vehicles and compare it with the original human-driven vehicle trajectories. All vehicles' initial conditions and the surrounding vehicles' movements are from the I-80 dataset. The merging vehicles' movements are controlled by the PCPG (purple vehicle in Figure \ref{I80_PCPG}) and by the human drivers as recorded (yellow vehicle in Figure \ref{I80_PCPG}), respectively. In the PCPG, a limited number of surrounding vehicles within the field of view are selected as game players, according to their time-to-colision (TTC) with the ego vehicle \cite{mine_1}. Specifically, a vehicle $j$ is considered as a game player if the difference between its TTC (denoted as $T_{j1}^c$) and the ego vehicle TTC ($T_{1j}^c$) is within a threshold, i.e., $|T_{j1}^c-T_{1j}^c|\leq \Delta T^c$. This criterion selects the most ``dangerous" vehicles for the ego vehicle. 
\begin{figure}[thpb]
\centering
\includegraphics[width=0.35\textwidth]{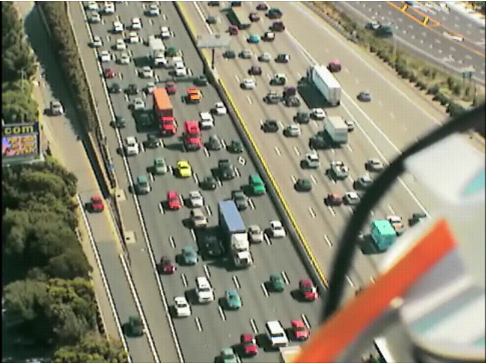}
\caption{A snapshot of I-80 traffic video \cite{I80}.}\label{I80}
\end{figure}

The merging vehicle performance is characterized by two metrics: The minimum longitudinal distance with the surrounding vehicles and the average ego vehicle speed during merging. The first metric indicates how safe this merge is, as longer inter-vehicle distances are safer. The second metric suggests how efficient the merge is, as higher  speed corresponds to higher efficiency. The validation results are reported in Table \ref{Table II}, which leads to the following observations. 
\begin{enumerate}
    \item In all tested situations, both the PCPG-driven and the human-driven vehicles can successfully merge into the highway, validating the effectiveness of the PCPG.
    \item In safety-critical situations, where the minimum inter-vehicle distance is less than $40m$, the PCPG is safer than the corresponding human drivers, as reflected by longer inter-vehicle distances.
    \item The average speeds of the human-driven and of the PCPG-driven vehicles are almost the same, validating the travel efficiency of the PCPG. 
\end{enumerate}}

 \begin{figure}[thpb]
\centering
\includegraphics[width=0.42\textwidth]{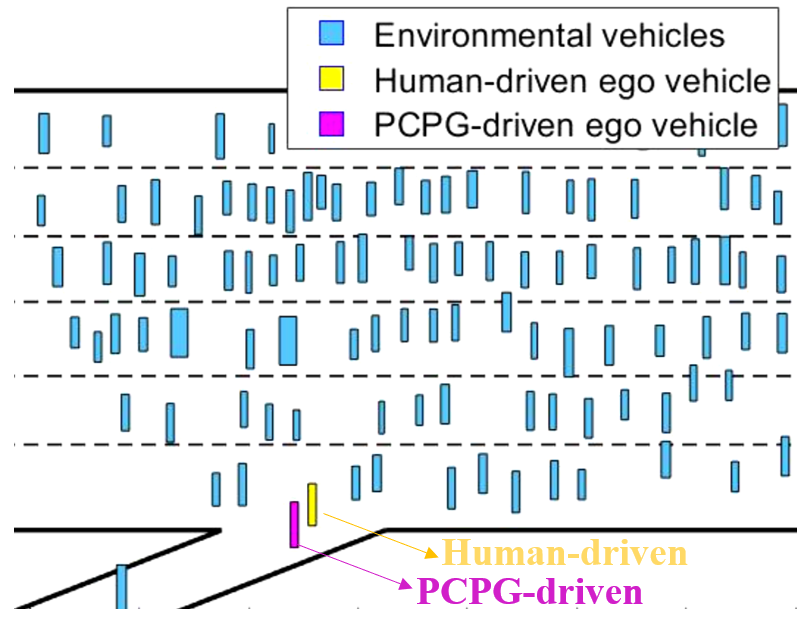}
\caption{A snapshot of the validation simulation video, where the merging vehicle trajectories are from the PCPG (purple) and from the original human driver (yellow). 
}\label{I80_PCPG}
\end{figure}
\begin{table}[!h]
\centering
\caption{Statistics of merging behaviors}\label{Table II}
\begin{tabular}{c|c|c|c|c}
\hline  
\multicolumn{1}{c|}{\multirow{2}{*}{}}
&\multicolumn{2}{c|}{\multirow{2}{*}{Min. inter-vehicle distance}}
&\multicolumn{2}{c}{\multirow{2}{*}{Ave. ego vehicle speed}}\\
\multicolumn{1}{c|}{\multirow{2}{*}{}}
&\multicolumn{2}{c|}{\multirow{2}{*}{}}
&\multicolumn{2}{c}{\multirow{2}{*}{}}
\\
\hline
Vehicle & Human driver & PCPG & Human driver & PCPG 
\\
\hline
\multirow{1}{*}{1}  & $12m$ & $16m$ & $13m/s$ & $13m/s$
\\
\hline
\multirow{1}{*}{2} & $21m$ & $29m$ & $13m/s$ & $13m/s$
\\
\hline
\multirow{1}{*}{3} & $85m$ & $79m$ & $20m/s$ & $19m/s$
\\
\hline
\multirow{1}{*}{4} & $12m$ & $19m$ & $13m/s$ & $13m/s$
\\
\hline
\multirow{1}{*}{5} & $35m$ & $37m$ & $13m/s$ & $16m/s$
\\
\hline
\multirow{1}{*}{6} & $12m$ & $21m$ & $13m/s$ & $13m/s$
\\
\hline
\multirow{1}{*}{7} & $25m$ & $20m$ & $14m/s$ & $14m/s$
\\
\hline
\multirow{1}{*}{8} & $12m$ & $16m$ & $12m/s$ & $12m/s$
\\
\hline
\multirow{1}{*}{9} & $8m$ & $10m$ & $12m/s$ & $12m/s$
\\
\hline
\multirow{1}{*}{10} & $12m$ & $14m$ & $7m/s$ & $7m/s$
\\
\hline
\end{tabular}
\end{table}

\section{Conclusion} \label{conclusion}
In this paper, a predictor-corrector potential game framework has been proposed to address the AV decision-making problem.  A receding horizon multi-player game was formulated to mimic human driver reasoning and to characterize agent interactions. To ensure solution existence, algorithm convergence, and computational scalability, a potential game based Predictor is developed.  To account for inaccurate cost functions of surrounding agents, a best-response based Corrector is introduced. By feeding back the surrounding agents' action deviations to the ego vehicle decision-making, the prediction error on other agents' future trajectories is corrected, leading to improved ego vehicle strategies. This framework guarantees the ego vehicle safety under suitable conditions and approximates the optimal solution despite the lack of information.
Comparative studies between PG, PCPG and  CBF-based approaches show that the PCPG ensures the ego vehicle safety, while the CBF-based approach is not always able to do so if the ego vehicle has limited control authority. It is because the PCPG enables the ego vehicle to always keep in mind the global situation, while the CBF constraints become active only when a collision threat is clear, i.e., passively reacting to the environment instead of proactively predicting and planning.  
With this framework, the long-standing challenges in AV decision-making, including safety, interpretability,  computational scalability, applicability to diverse scenarios, and human-like  intelligence, are all appropriately addressed. As a future work, we will further improve the practicability of the PCPG framework by considering more advanced Corrector designs, e.g., filtered versions of previous action deviations, to account for possible sensing noises.

\bibliography{references}

\begin{thebibliography}{10}
\providecommand{\url}[1]{#1}
\csname url@samestyle\endcsname
\providecommand{\newblock}{\relax}
\providecommand{\bibinfo}[2]{#2}
\providecommand{\BIBentrySTDinterwordspacing}{\spaceskip=0pt\relax}
\providecommand{\BIBentryALTinterwordstretchfactor}{4}
\providecommand{\BIBentryALTinterwordspacing}{\spaceskip=\fontdimen2\font plus
\BIBentryALTinterwordstretchfactor\fontdimen3\font minus \fontdimen4\font\relax}
\providecommand{\BIBforeignlanguage}[2]{{%
\expandafter\ifx\csname l@#1\endcsname\relax
\typeout{** WARNING: IEEEtran.bst: No hyphenation pattern has been}%
\typeout{** loaded for the language `#1'. Using the pattern for}%
\typeout{** the default language instead.}%
\else
\language=\csname l@#1\endcsname
\fi
#2}}
\providecommand{\BIBdecl}{\relax}
\BIBdecl

\bibitem{online2}
I.~GNSS, ``What is the future of autonomous vehicles?'' https://insidegnss.com/q-what-is-the-future-of-autonomous-vehicles/, Tech. Rep., 2022.

\bibitem{online3}
CNBC, ``Where the billions spent on autonomous vehicles by u.s. and chinese giants is heading,'' https://www.cnbc.com/2022/05/21/why-the-first-autonomous-vehicles-winners-wont-be-in-your-driveway.html, Tech. Rep., 2022.

\bibitem{RSS}
S.~Shalev-Shwartz, S.~Shammah, and A.~Shashua, ``On a formal model of safe and scalable self-driving cars,'' \emph{arXiv preprint arXiv:1708.06374}, 2017.

\bibitem{cbf}
A.~D. Ames, X.~Xu, J.~W. Grizzle, and P.~Tabuada, ``Control barrier function based quadratic programs for safety critical systems,'' \emph{IEEE Transactions on Automatic Control}, vol.~62, no.~8, pp. 3861--3876, 2016.

\bibitem{logical_1}
D.~D. Salvucci, E.~R. Boer, and A.~Liu, ``Toward an integrated model of driver behavior in cognitive architecture,'' \emph{Transportation Research Record}, vol. 1779, no.~1, pp. 9--16, 2001.

\bibitem{logical_2}
P.~Hidas, ``Modelling lane changing and merging in microscopic traffic simulation,'' \emph{Transportation Research Part C: Emerging Technologies}, vol.~10, no. 5-6, pp. 351--371, 2002.

\bibitem{MDP}
S.~Brechtel, T.~Gindele, and R.~Dillmann, ``Probabilistic decision-making under uncertainty for autonomous driving using continuous pomdps,'' in \emph{17th international IEEE Conference on Intelligent Transportation Systems (ITSC)}, 2014, pp. 392--399.

\bibitem{fuzzy}
J.~Perez, V.~Milanes, E.~Onieva, J.~Godoy, and J.~Alonso, ``Longitudinal fuzzy control for autonomous overtaking,'' in \emph{IEEE International Conference on Mechatronics}, 2011, pp. 188--193.

\bibitem{rss_conservative}
S.~Liu, X.~Wang, O.~Hassanin, X.~Xu, M.~Yang, D.~Hurwitz, and X.~Wu, ``Calibration and evaluation of responsibility-sensitive safety (rss) in automated vehicle performance during cut-in scenarios,'' \emph{Transportation research part C: emerging technologies}, vol. 125, p. 103037, 2021.

\bibitem{RSS_parameter}
H.~K{\"o}nigshof, F.~Oboril, K.-U. Scholl, and C.~Stiller, ``A parameter analysis on rss in overtaking situations on german highways,'' in \emph{IEEE Intelligent Vehicles Symposium (IV)}, 2022, pp. 1081--1086.

\bibitem{applicability}
S.~Kuutti, R.~Bowden, Y.~Jin, P.~Barber, and S.~Fallah, ``A survey of deep learning applications to autonomous vehicle control,'' \emph{IEEE Transactions on Intelligent Transportation Systems}, vol.~22, no.~2, pp. 712--733, 2020.

\bibitem{Ford_1}
S.~Nageshrao, H.~E. Tseng, and D.~Filev, ``Autonomous highway driving using deep reinforcement learning,'' in \emph{IEEE International Conference on Systems, Man and Cybernetics (SMC)}, 2019, pp. 2326--2331.

\bibitem{14}
X.~Xu, L.~Zuo, X.~Li, L.~Qian, J.~Ren, and Z.~Sun, ``A reinforcement learning approach to autonomous decision making of intelligent vehicles on highways,'' \emph{IEEE Transactions on Systems, Man, and Cybernetics: Systems}, vol.~50, no.~10, pp. 3884--3897, 2018.

\bibitem{15}
D.~C.~K. Ngai and N.~H.~C. Yung, ``A multiple-goal reinforcement learning method for complex vehicle overtaking maneuvers,'' \emph{IEEE Transactions on Intelligent Transportation Systems}, vol.~12, no.~2, pp. 509--522, 2011.

\bibitem{humanlike_learning}
L.~Li, K.~Ota, and M.~Dong, ``Humanlike driving: Empirical decision-making system for autonomous vehicles,'' \emph{IEEE Transactions on Vehicular Technology}, vol.~67, no.~8, pp. 6814--6823, 2018.

\bibitem{endtoend}
H.~Xu, Y.~Gao, F.~Yu, and T.~Darrell, ``End-to-end learning of driving models from large-scale video datasets,'' in \emph{Proceedings of the IEEE conference on computer vision and pattern recognition}, 2017, pp. 2174--2182.

\bibitem{endtoend2}
Y.~Xiao, F.~Codevilla, A.~Gurram, O.~Urfalioglu, and A.~M. L{\'o}pez, ``Multimodal end-to-end autonomous driving,'' \emph{IEEE Transactions on Intelligent Transportation Systems}, 2020.

\bibitem{explainable_2}
M.~Du, N.~Liu, and X.~Hu, ``Techniques for interpretable machine learning,'' \emph{Communications of the ACM}, vol.~63, no.~1, pp. 68--77, 2019.

\bibitem{explainable_1}
L.~Wells and T.~Bednarz, ``Explainable ai and reinforcement learning—a systematic review of current approaches and trends,'' \emph{Frontiers in artificial intelligence}, vol.~4, p. 550030, 2021.

\bibitem{human}
P.~Hang, C.~Lv, Y.~Xing, C.~Huang, and Z.~Hu, ``Human-like decision making for autonomous driving: A noncooperative game theoretic approach,'' \emph{IEEE Transactions on Intelligent Transportation Systems}, vol.~22, no.~4, pp. 2076--2087, 2020.

\bibitem{my_game}
M.~Liu, Y.~Wan, F.~L. Lewis, and V.~G. Lopez, ``Adaptive optimal control for stochastic multiplayer differential games using on-policy and off-policy reinforcement learning,'' \emph{IEEE transactions on neural networks and learning systems}, vol.~31, no.~12, pp. 5522--5533, 2020.

\bibitem{my_graphicalgame}
M.~Liu, Y.~Wan, V.~G. Lopez, F.~L. Lewis, G.~A. Hewer, and K.~Estabridis, ``Differential graphical game with distributed global nash solution,'' \emph{IEEE Transactions on Control of Network Systems}, vol.~8, no.~3, pp. 1371--1382, 2021.

\bibitem{Victor}
V.~G. Lopez, F.~L. Lewis, M.~Liu, Y.~Wan, S.~Nageshrao, and D.~Filev, ``Game-theoretic lane-changing decision making and payoff learning for autonomous vehicles,'' \emph{IEEE Transactions on Vehicular Technology}, vol.~71, no.~4, pp. 3609--3620, 2022.

\bibitem{mine_1}
M.~Liu, Y.~Wan, F.~Lewis, S.~Nageshrao, and D.~Filev, ``A three-level game-theoretic decision-making framework for autonomous vehicles,'' \emph{IEEE Transactions on Intelligent Transportation Systems}, 2022.

\bibitem{online_payoff}
Q.~Zhang, R.~Langari, H.~E. Tseng, D.~Filev, S.~Szwabowski, and S.~Coskun, ``A game theoretic model predictive controller with aggressiveness estimation for mandatory lane change,'' \emph{IEEE Transactions on Intelligent Vehicles}, vol.~5, no.~1, pp. 75--89, 2019.

\bibitem{supervisedlearning}
D.~Silver, J.~A. Bagnell, and A.~Stentz, ``Learning autonomous driving styles and maneuvers from expert demonstration,'' in \emph{Experimental Robotics}.\hskip 1em plus 0.5em minus 0.4em\relax Springer, 2013, pp. 371--386.

\bibitem{inverse_RL}
J.~Liu, L.~N. Boyle, and A.~Banerjee, ``An inverse reinforcement learning approach for customizing automated lane change systems,'' \emph{IEEE Transactions on Vehicular Technology}, 2022.

\bibitem{RL_VFA}
B.~R. Kiran, I.~Sobh, V.~Talpaert, P.~Mannion, A.~A. Al~Sallab, S.~Yogamani, and P.~P{\'e}rez, ``Deep reinforcement learning for autonomous driving: A survey,'' \emph{IEEE Transactions on Intelligent Transportation Systems}, 2021.

\bibitem{deadlock}
R.~Mandiau, A.~Champion, J.-M. Auberlet, S.~Espi{\'e}, and C.~Kolski, ``Behaviour based on decision matrices for a coordination between agents in a urban traffic simulation,'' \emph{Applied Intelligence}, vol.~28, no.~2, pp. 121--138, 2008.

\bibitem{dimitar_dynamics_2}
S.~Hong, J.~Lu, and D.~P. Filev, ``Driving behavior evaluation for future mobility: Application of online transition probability estimation,'' \emph{IEEE Transactions on Intelligent Transportation Systems}, 2019.

\bibitem{dimitar_dynamics_3}
J.~Lu, D.~Filev, and F.~Tseng, ``Real-time determination of driver's driving behavior during car following,'' \emph{SAE International Journal of Passenger Cars-Electronic and Electrical Systems}, vol.~8, no. 2015-01-0297, pp. 371--378, 2015.

\bibitem{sov}
V.-A. Le and A.~A. Malikopoulos, ``A cooperative optimal control framework for connected and automated vehicles in mixed traffic using social value orientation,'' \emph{arXiv preprint arXiv:2203.17106}, 2022.

\bibitem{closedloop1}
W.~Schwarting, A.~Pierson, J.~Alonso-Mora, S.~Karaman, and D.~Rus, ``Social behavior for autonomous vehicles,'' \emph{Proceedings of the National Academy of Sciences}, vol. 116, no.~50, pp. 24\,972--24\,978, 2019.

\bibitem{book_lf}
M.~O'Searcoid, \emph{Metric spaces}.\hskip 1em plus 0.5em minus 0.4em\relax Springer Science \& Business Media, 2006.

\bibitem{game_book}
Y.~Shoham and K.~Leyton-Brown, \emph{Multiagent systems: Algorithmic, game-theoretic, and logical foundations}.\hskip 1em plus 0.5em minus 0.4em\relax Cambridge University Press, 2008.

\bibitem{my_potential}
M.~Liu, I.~Kolmanovsky, H.~E. Tseng, S.~Huang, D.~Filev, and A.~Girard, ``Potential game-based decision-making for autonomous driving,'' \emph{IEEE Transactions on Intelligent Transportation Systems}, 2023.

\bibitem{potential_book}
Q.~D. L{\~a}, Y.~H. Chew, and B.-H. Soong, \emph{Potential Game Theory}.\hskip 1em plus 0.5em minus 0.4em\relax Springer, 2016.

\bibitem{tesla}
Tesla, ``Tesla ai day 2022,'' \url{https://www.youtube.com/watch?v=ODSJsviD_SU}.

\bibitem{add_dynamics_RL}
C.-J. Hoel, K.~Driggs-Campbell, K.~Wolff, L.~Laine, and M.~J. Kochenderfer, ``Combining planning and deep reinforcement learning in tactical decision making for autonomous driving,'' \emph{IEEE Transactions on Intelligent Vehicles}, vol.~5, no.~2, pp. 294--305, 2019.

\bibitem{add_dynamics_nan}
N.~Li, H.~Chen, I.~Kolmanovsky, and A.~Girard, ``An explicit decision tree approach for automated driving,'' in \emph{Dynamic Systems and Control Conference}, vol. 58271, 2017, p. V001T45A003.

\bibitem{add_dynamics_RTD}
S.~Kousik, S.~Vaskov, M.~Johnson-Roberson, and R.~Vasudevan, ``Safe trajectory synthesis for autonomous driving in unforeseen environments,'' in \emph{Dynamic Systems and Control Conference}, vol. 58271, 2017, p. V001T44A005.

\bibitem{bicycle}
J.~Kong, M.~Pfeiffer, G.~Schildbach, and F.~Borrelli, ``Kinematic and dynamic vehicle models for autonomous driving control design,'' in \emph{Proceedings of IEEE Intelligent Vehicles Symposium (IV)}, 2015, pp. 1094--1099.

\bibitem{validation_1}
N.~Ratliff, B.~Ziebart, K.~Peterson, J.~A. Bagnell, M.~Hebert, A.~K. Dey, and S.~Srinivasa, ``Inverse optimal heuristic control for imitation learning,'' in \emph{Artificial intelligence and statistics}, 2009, pp. 424--431.

\bibitem{validation_2}
B.~D. Ziebart, A.~L. Maas, J.~A. Bagnell, A.~K. Dey \emph{et~al.}, ``Maximum entropy inverse reinforcement learning.'' in \emph{Aaai}, vol.~8, 2008, pp. 1433--1438.

\bibitem{validation_3}
B.~D. Ziebart, A.~L. Maas, A.~K. Dey, and J.~A. Bagnell, ``Navigate like a cabbie: Probabilistic reasoning from observed context-aware behavior,'' in \emph{Proceedings of the 10th international conference on Ubiquitous computing}, 2008, pp. 322--331.

\bibitem{validation_4}
O.~Siebinga, A.~Zgonnikov, and D.~Abbink, ``A human factors approach to validating driver models for interaction-aware automated vehicles,'' \emph{ACM Transactions on Human-Robot Interaction (THRI)}, vol.~11, no.~4, pp. 1--21, 2022.

\bibitem{payoff1}
Q.~Dai, D.~Shen, J.~Wang, S.~Huang, and D.~Filev, ``Calibration of human driving behavior and preference using vehicle trajectory data,'' \emph{Transportation research part C: emerging technologies}, vol. 145, p. 103916, 2022.

\bibitem{pcca}
M.~Santillo and M.~Jankovic, ``Collision free navigation with interacting, non-communicating obstacles,'' in \emph{American Control Conference (ACC)}, 2021, pp. 1637--1643.

\bibitem{payoff2}
C.~Hubmann, M.~Becker, D.~Althoff, D.~Lenz, and C.~Stiller, ``Decision making for autonomous driving considering interaction and uncertain prediction of surrounding vehicles,'' in \emph{IEEE Intelligent Vehicles Symposium}, 2017, pp. 1671--1678.

\bibitem{payoff3}
J.~F. Fisac, E.~Bronstein, E.~Stefansson, D.~Sadigh, S.~S. Sastry, and A.~D. Dragan, ``Hierarchical game-theoretic planning for autonomous vehicles,'' in \emph{International Conference on Robotics and Automation (ICRA)}, 2019, pp. 9590--9596.

\bibitem{dynamics}
R.~Tian, N.~Li, I.~Kolmanovsky, Y.~Yildiz, and A.~R. Girard, ``Game-theoretic modeling of traffic in unsignalized intersection network for autonomous vehicle control verification and validation,'' \emph{IEEE Transactions on Intelligent Transportation Systems}, 2020.

\bibitem{PCPG_add1}
C.~Cenedese, M.~Cucuzzella, J.~Scherpen, S.~Grammatico, and M.~Cao, ``Highway traffic control via smart e-mobility--part i: Theory,'' \emph{arXiv preprint arXiv:2102.09354}, 2021.

\bibitem{PCPG_add2}
F.~Fabiani and S.~Grammatico, ``Multi-vehicle automated driving as a generalized mixed-integer potential game,'' \emph{IEEE Transactions on Intelligent Transportation Systems}, vol.~21, no.~3, pp. 1064--1073, 2019.

\bibitem{BR}
S.~Durand and B.~Gaujal, ``Complexity and optimality of the best response algorithm in random potential games,'' in \emph{International Symposium on Algorithmic Game Theory}.\hskip 1em plus 0.5em minus 0.4em\relax Springer, 2016, pp. 40--51.

\bibitem{jerk_limit}
F.~Feng, S.~Bao, J.~R. Sayer, C.~Flannagan, M.~Manser, and R.~Wunderlich, ``Can vehicle longitudinal jerk be used to identify aggressive drivers? an examination using naturalistic driving data,'' \emph{Accident Analysis \& Prevention}, vol. 104, pp. 125--136, 2017.

\bibitem{six_sigma}
D.~C. Montgomery and W.~H. Woodall, ``An overview of six sigma,'' \emph{International Statistical Review/Revue Internationale de Statistique}, pp. 329--346, 2008.

\bibitem{suzhou}
Q.~Dai, X.~Xu, W.~Guo, S.~Huang, and D.~Filev, ``Towards a systematic computational framework for modeling multi-agent decision-making at micro level for smart vehicles in a smart world,'' \emph{Robotics and Autonomous Systems}, vol. 144, p. 103859, 2021.

\bibitem{openloop1}
S.~Coskun, Q.~Zhang, and R.~Langari, ``Receding horizon markov game autonomous driving strategy,'' in \emph{2019 American Control Conference (ACC)}.\hskip 1em plus 0.5em minus 0.4em\relax IEEE, 2019, pp. 1367--1374.

\bibitem{IPM1}
F.~A. Potra and S.~J. Wright, ``Interior-point methods,'' \emph{Journal of computational and applied mathematics}, vol. 124, no. 1-2, pp. 281--302, 2000.

\bibitem{IPM2}
M.~Wright, ``The interior-point revolution in optimization: history, recent developments, and lasting consequences,'' \emph{Bulletin of the American mathematical society}, vol.~42, no.~1, pp. 39--56, 2005.

\bibitem{RTD}
S.~Kousik, S.~Vaskov, F.~Bu, M.~Johnson-Roberson, and R.~Vasudevan, ``Bridging the gap between safety and real-time performance in receding-horizon trajectory design for mobile robots,'' \emph{The International Journal of Robotics Research}, vol.~39, no.~12, pp. 1419--1469, 2020.

\bibitem{kaiwen}
K.~Liu, N.~Li, H.~E. Tseng, I.~Kolmanovsky, and A.~Girard, ``Interaction-aware trajectory prediction and planning for autonomous vehicles in forced merge scenarios,'' \emph{IEEE Transactions on Intelligent Transportation Systems}, vol.~24, no.~1, pp. 474--488, 2022.

\bibitem{add_T}
L.~Yang, X.~Li, W.~Guan, H.~M. Zhang, and L.~Fan, ``Effect of traffic density on drivers’ lane change and overtaking maneuvers in freeway situation—a driving simulator--based study,'' \emph{Traffic injury prevention}, vol.~19, no.~6, pp. 594--600, 2018.

\bibitem{comfort}
S.~Dominguez, A.~Ali, G.~Garcia, and P.~Martinet, ``Comparison of lateral controllers for autonomous vehicle: Experimental results,'' in \emph{19th International Conference on Intelligent Transportation Systems (ITSC)}, 2016, pp. 1418--1423.

\bibitem{ga}
MathWorks, ``Find minimum of function using genetic algorithm,'' \url{https://www.mathworks.com/help/gads/ga.html}.

\bibitem{qp}
------, ``Quadratic programming,'' \url{https://www.mathworks.com/help/optim/ug/quadprog.html}.

\bibitem{I80}
U.~D. of~Transportation Federal Highway~Administration, ``Next generation simulation (ngsim) program i-80 videos. [dataset]. provided by its datahub through data.transportation.gov.'' Accessed 2023-08-12 from \url{http://doi.org/10.21949/1504477}.

\bibitem{zhaojian}
M.~R. Hajidavalloo, Z.~Li, D.~Chen, A.~Louati, S.~Feng, and W.~B. Qin, ``Mechanical system inspired microscopic traffic model: Modeling, analysis, and validation,'' \emph{IEEE Transactions on Intelligent Vehicles}, vol.~8, no.~1, pp. 301--312, 2022.

\end{thebibliography}
\bibliographystyle{IEEEtran}

\begin{IEEEbiography}[{\includegraphics[width=1in,height=1.25in,clip,keepaspectratio]{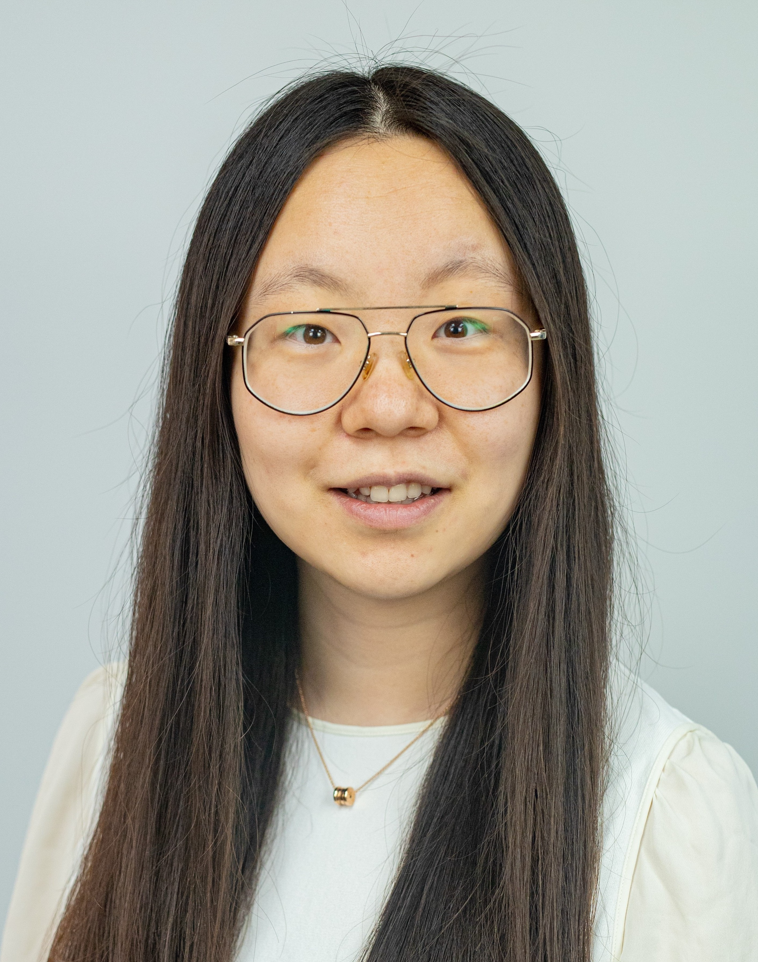}}]{Mushuang Liu} is an Assistant Professor in the Department of Mechanical and Aerospace Engineering at the University of Missouri, Columbia, MO. She worked as a postdoc in the Department of Aerospace Engineering at the University of Michigan, Ann Arbor, MI. She received her Ph.D degree from the University of Texas at Arlington in 2020 and her B.S. degree from the University of Electronic Science and Technology of China in 2016. Her research lies in  control,  learning, and games for multi-agent systems.
\end{IEEEbiography}

\begin{IEEEbiography}[{\includegraphics[width=1in,height=1.25in,clip,keepaspectratio]{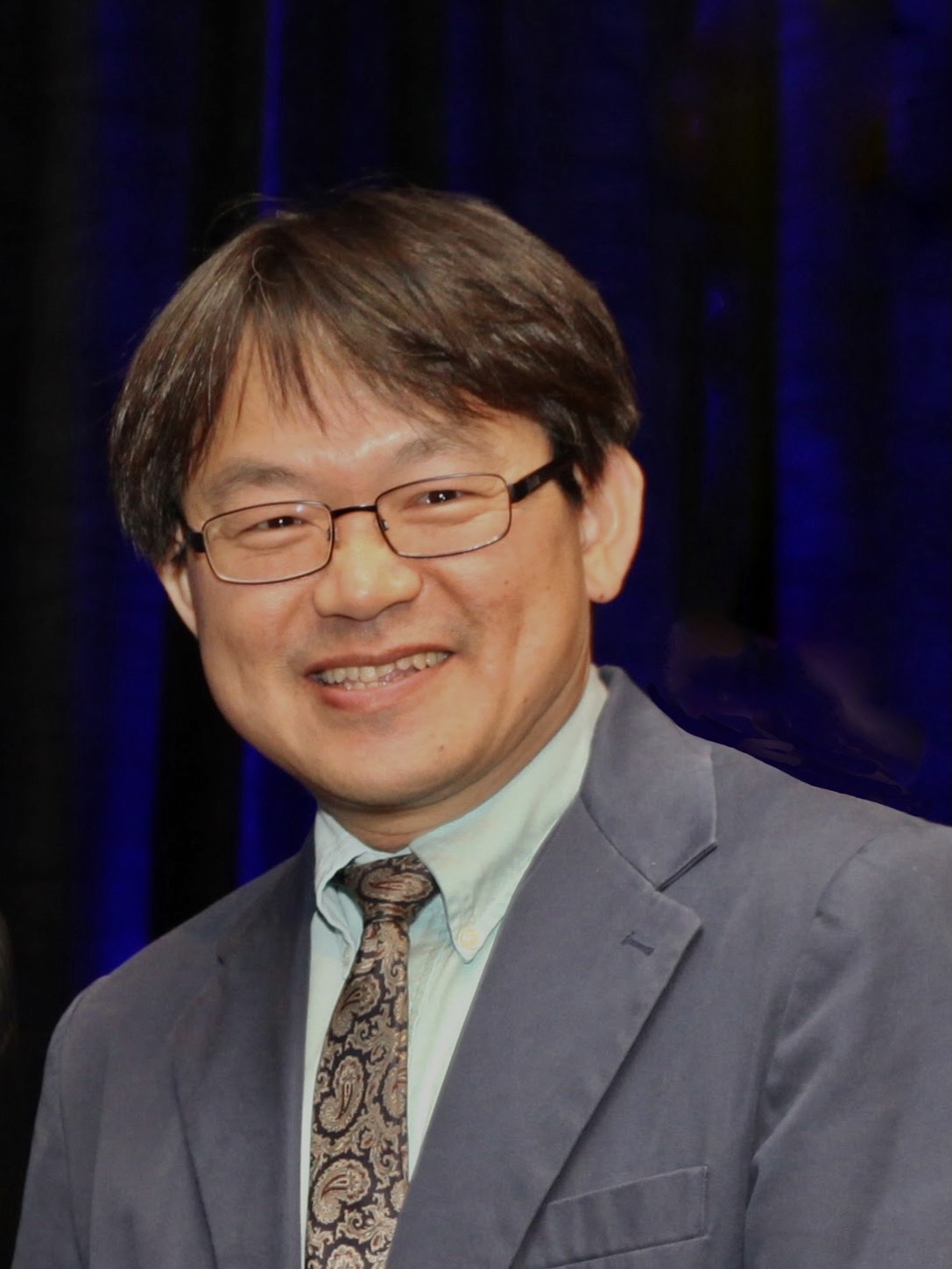}}]{H. Eric Tseng}  received the B.S. degree from the National Taiwan University, Taipei, Taiwan, in 1986, and the M.S. and Ph.D. degrees in mechanical engineering from the University of California at Berkeley, Berkeley, in 1991 and 1994, respectively. In 1994, he joined Ford Motor Company. At Ford, he is currently a Senior Technical Leader of Controls and Automated Systems in Research and Advanced Engineering. Many of his contributed technologies led to production vehicles implementation. His technical achievements have been recognized internally seven times with Ford’s highest technical award—the Henry Ford Technology Award, as well as externally by the American Automatic Control Council with Control Engineering Practice Award in 2013. He has over 100 U.S. patents
and over 120 publications. He is an NAE Member.
\end{IEEEbiography}

\begin{IEEEbiography}[{\includegraphics[width=1in,height=1.5in,clip,keepaspectratio]{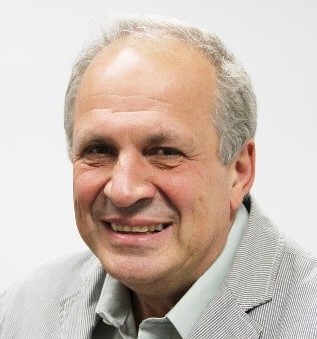}}]{Dimitar Filev} (Fellow, IEEE) is Senior Henry Ford Technical Fellow in Control and AI with Research $\&$ Advanced Engineering – Ford Motor Company. His research is in computational intelligence, AI and intelligent control, and their applications to autonomous driving, vehicle systems, and automotive engineering.  He holds over 100 granted US patents and has been awarded with the IEEE SMCS 2008 Norbert Wiener Award and the 2015 Computational Intelligence Pioneer’s Award. Dr. Filev is a Fellow of the IEEE and a member of the National Academy of Engineering. He was President of the IEEE Systems, Man, and Cybernetics Society (2016-2017).
\end{IEEEbiography}

\begin{IEEEbiography}[{\includegraphics[width=1in,height=1.5in,clip,keepaspectratio]{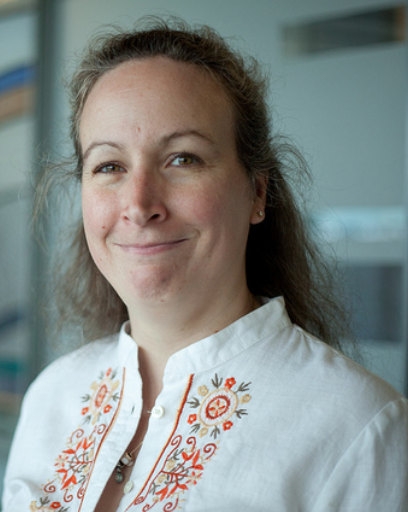}}]{Anouck Girard} received
the Ph.D. degree in ocean engineering from the
University of California at Berkeley, Berkeley, CA, USA, in 2002. She has been with the University of Michigan, Ann Arbor, MI, USA, since 2006, where she is currently a Professor of aerospace engineering. She has coauthored the book Fundamentals of Aerospace Navigation and Guidance (Cambridge University Press, 2014). Her current research interests include vehicle dynamics and control systems. She was a recipient of the Silver Shaft Teaching Award from the University of Michigan and the Best Student Paper Award from the American Society of Mechanical Engineers.
\end{IEEEbiography}

\begin{IEEEbiography}[{\includegraphics[width=1in,height=1.25in,clip,keepaspectratio]{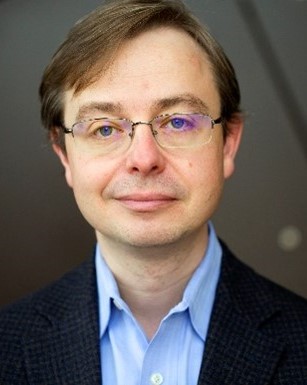}}]{Ilya Kolmanovsky} (Fellow, IEEE) received the Ph.D. degree in aerospace engineering from the University of Michigan, Ann Arbor, MI, USA, in 1995. Prior to joining the University of Michigan as a Faculty Member in 2010, he was with Ford Research and Advanced Engineering, Dearborn, MI, for close to 15 years. He is currently a Pierre T. Kabamba Collegiate Professor in the Department of Aerospace Engineering at the University of Michigan. His research interests include control theory for systems with state and control constraints, and control applications to aerospace and automotive systems. He is a Fellow of IFAC and NAI and a Senior Editor of IEEE TRANSACTIONS ON CONTROL SYSTEMS TECHNOLOGY.
\end{IEEEbiography}

\end{document}